\documentclass[11pt,reqno]{amsart}
\usepackage{txfonts}
\usepackage{mathrsfs}
\usepackage{amsfonts}
\allowdisplaybreaks[4]
\usepackage{graphicx} 
\usepackage{epsfig}
\usepackage{dsfont}
\usepackage{braket}
\usepackage{amssymb}
\usepackage{tcolorbox}
\usepackage{tikz}
\usepackage{amsmath,xypic}
\usepackage{cite,color}
\usepackage{enumerate}
\newtheorem{theorem}{Theorem}

\newtheorem{proposition}{Proposition}
\newtheorem{corollary}{Corollary}

\newtheorem{lemma}{Lemma}

\newtheorem{example}{Example}
\newcommand{\be}{\begin{equation}}
\newcommand{\ee}{\end{equation}}
\newcommand{\bea}{\begin{eqnarray}}
\newcommand{\eea}{\end{eqnarray}}
\newcommand{\ba}{\begin{array}}
	\newcommand{\ea}{\end{array}}
\newcommand{\bean}{\begin{eqnarray*}}
	\newcommand{\eean}{\end{eqnarray*}}

\begin{document}
\title{Tau functions of the constrained matrix KP hierarchy}
\author{Xiaohan Fan$^1$, Jipeng Cheng$^{2*}$, Jinbiao Wang$^1$}
\dedicatory { $^1$ School of Mathematics, China University of
Mining and Technology, Xuzhou, Jiangsu 221116, P.\ R.\ China\\
$^2$ School of Mathematical Sciences, Huaqiao University, Quanzhou, Fujian 362021, P.\ R.\ China}
\thanks{*Corresponding author. Email: chengjp@hqu.edu.cn, chengjipeng1983@163.com.}
\begin{abstract}
  The constrained matrix KP hierarchy $(L^{k})_{<0}=\sum_{i=1}^{m}Q_{i}\partial^{-1}R_{i}^{\intercal}$ is investigated from the aspects of tau functions. Firstly, the matrix KP
  hierarchy is viewed as one special reduction of the multi-component KP hierarchy. Then bilinear equations of the constrained matrix KP hierarchy as the multi-component KP hierarchy are given in terms of tau functions. Finally based upon these results, the tau functions for the constrained matrix KP hierarchy are constructed by using the multi-component boson-fermion correspondence. Notice that the solutions of the constrained matrix KP hierarchy are derived without using quasi-determinants.
\end{abstract}

\maketitle

\tableofcontents

\section{Introduction}
\subsection{The matrix KP hierarchy}
The KP hierarchy \cite{harnad2021,dickey2003,date1983,mulase1994} is one of the most important research objects in integrable systems, which has many famous extensions, such as the multi-component KP hierarchy \cite{kac2023,kac1993,teo2011}, the non-commutative KP hierarchy \cite{adler1998,gilson2007,blower2026,he2011,kupershmidt2000} and so on. And these extensions of the KP hierarchy constitute the famous KP theory\cite{harnad2021,dickey2003,date1983,mulase1994}, which is widely used in mathematical physics. Here in this paper, we are more interested in the union of the multi-component and the non-commutative extensions for the KP hierarchy, that is the matrix KP hierarchy \cite{prokofev2021,pashkov2018,tacchella2011,bergvelt2009,konopel1991,huang2011}.

The matrix KP hierarchy \cite{prokofev2021,pashkov2018,tacchella2011,bergvelt2009,konopel1991,huang2011} is defined by the following Lax equation:
\begin{equation}\label{matrixKPLaxeq}
  L(\mathbf{t},\partial)_{t_{n}}=[B_{n}(\mathbf{t},\partial),L(\mathbf{t},\partial)].
\end{equation}
Here the Lax operator $L(\mathbf{t},\partial)$ is given by the following matrix pseudo-differential operator
\begin{equation}\label{matrixLaxdefination}
  L(\mathbf{t},\partial)=I_{N}\partial+U_{2}(\mathbf{t})\partial^{-1}+U_{3}(\mathbf{t})\partial^{-2}+\cdots,
\end{equation}
where $\partial=\partial_{ x}$, \( I_N \) denotes the \( N \times N \) identity matrix, and each $U_i(\mathbf{t})$ with $\mathbf{t}=(t_{1}=x, t_{2}, t_{3},\cdots)$ is an $N\times N$ matrix of the following form:
\begin{align*}
U_i(\mathbf{t})=
\begin{pmatrix}
U_{i,11}(\mathbf{t}) & \cdots & U_{i,1N}(\mathbf{t}) \\
\vdots & \ddots & \vdots \\
U_{i,N1}(\mathbf{t}) & \cdots & U_{i,NN}(\mathbf{t})
\end{pmatrix},\quad
i\geq2.
\end{align*}
And $B_{n}(\mathbf{t},\partial)=L^{n}(\mathbf{t},\partial)_{\geq0}$ with $(\sum_{i}A_{i}\partial^{i})_{\geq0}=\sum_{i\geq0}A_{i}\partial^{i}. $

Similar to the ordinary KP hierarchy, the matrix KP Lax operator $L$ can also be expressed in terms of the matrix dressing operator \cite{konopel1991,prokofev2021,huang2011}
$$W(\mathbf{t},\partial)=I_{N}+W_{1}(\mathbf{t})\partial^{-1}+W_{2}(\mathbf{t})\partial^{-2}+\cdots,$$
where $W_{i}(\mathbf{t})$ is an $N\times N$ matrix function of $\mathbf{t}$ such that the following identities hold:
\begin{align*}
  L(\mathbf{t},\partial)=W(\mathbf{t},\partial)\partial W(\mathbf{t},\partial)^{-1},\quad W(\mathbf{t},\partial)_{t_{n}}=-L(\mathbf{t},\partial)^{n}_{<0}W(\mathbf{t},\partial),
\end{align*}
where $(\sum_{i}A_{i}\partial^{i})_{<0}=\sum_{i<0}A_{i}\partial^{i}.$
Furthermore, we can define the matrix wave function $\Psi(\mathbf{t},z)$ and the matrix adjoint wave function $\Psi^{*}(\mathbf{t},z)$ as follows \cite{prokofev2021,kac2023}:
$$\Psi(\mathbf{t},z)=W(\mathbf{t},\partial)(e^{\xi(\mathbf{t},z)}), \quad \Psi^{*}(\mathbf{t},z)=W(\mathbf{t},\partial)^{-1*}(e^{-\xi(\mathbf{t},z)}),$$
where $\xi(\mathbf{t},z)=\sum_{i=1}^{+\infty}t_{i}z^{i},$ and $\Psi(\mathbf{t},z), \Psi^{*}(\mathbf{t},z)$ are $N\times N$ matrices. The adjoint symbol $``*"$ is defined by $\left( \sum_{j} A_j \partial^j \right)^* = \sum_{j} (-\partial)^j A_j^{\intercal},$
with $A^{\intercal}$ being the transpose of the matrix $A$. Moreover, the matrix wave function $\Psi(\mathbf{t},z)$ and the matrix adjoint wave function $\Psi^{*}(\mathbf{t},z)$ satisfy:
\begin{align*}
  L(\mathbf{t},\partial)^n\bigl(\Psi(\mathbf{t},z)\bigr) &= z^n \Psi(\mathbf{t},z), \quad \Psi(\mathbf{t},z)_{\mathbf{t}_n} = B_n(\mathbf{t},\partial)\bigl(\Psi(\mathbf{t},z) \bigr),\\
  L(\mathbf{t},\partial)^{n*}\bigl(\Psi^{*}(\mathbf{t},z)\bigr) &= z^n \Psi^{*}(\mathbf{t},z), \quad \Psi^{*}(\mathbf{t},z)_{\mathbf{t}_n} = -B_n(\mathbf{t},\partial)^{*}\bigl(\Psi^{*}(\mathbf{t},z) \bigr),
\end{align*}
and the following bilinear equation \cite{kac2023,kac1993}:
\begin{align}\label{ordinarybilinearequation}
  \operatorname{Res}_{z}\Psi(\mathbf{t},z)\Psi^*(\mathbf{t}',z)^\intercal=0,
\end{align}
where $\operatorname{Res}_{z}\sum_{i}A_{i}z^{i}=A_{-1}$ for $N\times N$ matrices $A_{i}$.

In this paper, we are more interested in one important reduction of the matrix KP hierarchy, called the $(k,m)$-constrained matrix KP hierarchy \cite{chvartatskyi2013,huang2011,oevel1993,kundu1995}:
\begin{align}\label{con-matrix-KP-Lax-operator}
  L(\mathbf{t},\partial)^{k}=L(\mathbf{t},\partial)^{k}_{\geq0}+\sum_{i=1}^{m}Q_{i}(\mathbf{t})\partial^{-1}R_{i}(\mathbf{t})^{\intercal}, \quad 1\leq i\leq m,
\end{align}
where $Q_{i}$ and $R_{i}$ are $N\times N$ matrix functions of $\mathbf{t}$, satisfying the following identities:
\begin{align}\label{qrderivative}
   Q_{i,t_{n}}(\mathbf{t})=B_n(\mathbf{t},\partial)(Q_{i}(\mathbf{t})),\quad R_{i,t_{n}}(\mathbf{t})=-B_n(\mathbf{t},\partial)^{*}(R_{i}(\mathbf{t})).
\end{align}
Here the $(k,m)$-constrained matrix KP hierarchy refers to the system
consisting of \eqref{matrixKPLaxeq} \eqref{con-matrix-KP-Lax-operator} and \eqref{qrderivative}. The aim of this paper is to solve the
$(k,m)$-constrained matrix KP hierarchy.

\subsection{The multi-component KP hierarchy}
In this subsection, we will show that the matrix KP hierarchy can be viewed as the reduction of the multi-component KP hierarchy, so that we can discuss the matrix KP hierarchy by tau functions of multi-component KP hierarchy.

The $N$-component KP hierarchy \cite{kac2023,kac1993,teo2011,zhang1999} is defined by $N+1$ matrix pseudo-differential operators $\mathcal{L}(\mathbf{x},D)$ and $\mathcal{C}^{(\alpha)}(\mathbf{x},D)$ with the following forms:
\begin{align*}
  \mathcal{L}(\mathbf{x},D) = I_N D + \sum_{j=1}^{\infty} \widetilde{U}_j(\mathbf{x}) D^{-j}, \quad
   \mathcal{C}^{(\alpha)}(\mathbf{x},D) = E_{\alpha\alpha} + \sum_{j=1}^{\infty} C_j^{(\alpha)}(\mathbf{x}) D^{-j}, \quad 1\leq \alpha\leq N,
\end{align*}
where $\mathbf{x}=(\mathbf{x}^{(1)},\mathbf{x}^{(2)},\cdots,\mathbf{x}^{(N)}),\,\mathbf{x}^{(\alpha)}=(x^{(\alpha)}_{1},x^{(\alpha)}_{2},\cdots)$ and \( D = \sum_{\alpha=1}^{N}\partial_{x_1^{(\alpha)} } \), $I_{N}$ is the $N\times N$ unit matrix, and \( E_{ij} \) stands for the \( N \times N \) matrix whose $(i,j)$ entry is $1$ and all other entries are zero.
Meanwhile, $\mathcal{L}$ and $C^{(\alpha)}$ also satisfy the following Lax equation:
\begin{equation}
\begin{aligned}\label{LandC}
&\mathcal{L}_{x_n^{(\beta)}} = \left[ \mathcal{B}_n^{(\beta)}, \mathcal{L} \right], \quad \mathcal{C}^{(\alpha)}_{x_n^{(\beta)}}= \left[ \mathcal{B}_n^{(\beta)}, \mathcal{C}^{(\alpha)} \right], \quad \mathcal{B}_n^{(\alpha)}=(\mathcal{C}^{(\alpha)}\mathcal{L}^{n})_{\geq0}, \quad 1\leq \alpha, \beta\leq N, \\
&\sum_{\alpha=1}^{N} \mathcal{C}^{(\alpha)} = I_N, \quad \mathcal{C}^{(\alpha)} \mathcal{L} = \mathcal{L} \mathcal{C}^{(\alpha)}, \quad \mathcal{C}^{(\alpha)} \mathcal{C}^{(\beta)} = \delta_{\alpha\beta} \mathcal{C}^{(\alpha)},
\end{aligned}
\end{equation}
where $(\sum_{i}A_{i}D^{i})_{\geq0}=\sum_{i\geq0}A_{i}D^{i}. $

The matrix KP hierarchy can be obtained from the $N$-component KP hierarchy by restricting the time variables in the following way \cite{zabrodin2019,prokofev2021}:
$$
\mathbf{x}^{(\alpha)} = \mathbf{t} \quad \text{for all } \alpha \, \text{with} \, 1\leq\alpha\leq N.
$$
In this case, $\partial_{t_{n}}=\sum_{\alpha=1}^{N}\partial_{x^{(\alpha)}_{n}}$ and in particular $\partial=D$. Thus if we set $L(\mathbf{t},\partial)=\mathcal{L}(\mathbf{x},D)|_{\mathbf{x}^{(\alpha)} = \mathbf{t}, 1\leq \alpha\leq N}$, then
\begin{align*}
L(\mathbf{t},\partial)_{t_{n}}&=\sum_{\beta=1}^N \mathcal{L}_{x^{(\beta)}_{n}}(\mathbf{x},D)|_{\mathbf{x}^{(\alpha)} = \mathbf{t}, 1\leq\alpha\leq N}=\sum_{\beta=1}^{N}[ \mathcal{B}_n^{(\beta)}(\mathbf{x},D), \mathcal{L}(\mathbf{x},D)]|_{\mathbf{x}^{(\alpha)} = \mathbf{t}, 1\leq\alpha\leq N}\\
 &=\bigg[\bigg(\mathcal{L}^{n}(\mathbf{x},D)\sum_{\beta=1}^{N}C^{(\beta)}(\mathbf{x},D)\bigg)_{\geq0}, \mathcal{L}(\mathbf{x},D)\bigg]\bigg|_{\mathbf{x}^{(\alpha)}=\mathbf{t}, 1\leq\alpha\leq N}=\bigg[\mathcal{L}^{n}(\mathbf{x},D)_{\geq0}, \mathcal{L}(\mathbf{x},D)\bigg]\bigg|_{\mathbf{x}^{(\alpha)}=\mathbf{t}, 1\leq\alpha\leq N}\\
 &=[B_{n}(\mathbf{t},\partial),L(\mathbf{t},\partial)],
\end{align*}
which is just the Lax equation of the matrix KP hierarchy. So the matrix KP hierarchy is a special reduction of the multi-component KP hierarchy.

Similarly for the $N$-component KP hierarchy, there exists a matrix dressing operator \cite{kac2023,kac1993} $$\mathcal{W}(\mathbf{x},D)= I_N + \sum_{i=1}^{\infty} \mathcal{W}_i(\mathbf{x}) D^{-i},$$
such that
\begin{align*}
&\mathcal{L}(\mathbf{x},D)= \mathcal{W}(\mathbf{x},D) D \mathcal{W}(\mathbf{x},D)^{-1}  ,\quad \mathcal{C}^{(\alpha)}(\mathbf{x},D)= \mathcal{W}(\mathbf{x},D) E_{\alpha\alpha} \mathcal{W}(\mathbf{x},D)^{-1}, \\
&\mathcal{W}(\mathbf{x},D)_{x_n^{(\alpha)}} = -\left( \mathcal{W}(\mathbf{x},D) E_{\alpha\alpha} D^n \mathcal{W}(\mathbf{x},D)^{-1} \right)_{<0} \mathcal{W}(\mathbf{x},D), \quad 1\leq \alpha\leq N, \,n\geq1,
\end{align*}
where $(\sum_{i}A_{i}D^{i})_{<0}=\sum_{i<0}A_{i}D^{i}. $
If we define the matrix wave function  \( \widetilde{\Psi}(\mathbf{x},z) \) and the matrix adjoint wave function  \( \widetilde{\Psi}^{*}(\mathbf{x},z) \) of the $N$-component KP hierarchy as follows:
\begin{align*}
  &\widetilde{\Psi}(\mathbf{x},z) =\left(\widetilde{\Psi}_{\alpha\beta}(\mathbf{x}, z)\right)_{\alpha,\beta=1}^N \triangleq \mathcal{W}(\mathbf{x},D)\left( \sum_{\alpha=1}^N E_{\alpha\alpha} e^{\xi(\mathbf{x}^{(\alpha)},z)} \right),\\
&\widetilde{\Psi}^{*}(\mathbf{x},z) =\left(\widetilde{\Psi}^{*}_{\alpha\beta}(\mathbf{x}, z)\right)_{\alpha,\beta=1}^N \triangleq {{(\mathcal{W}(\mathbf{x},D)^{-1})^{*}}}\left( \sum_{\alpha=1}^N E_{\alpha\alpha} e^{-\xi(\mathbf{x}^{(\alpha)},z)} \right),
\end{align*}
where $\xi(\mathbf{x}^{(\alpha)},z) = \sum_{l=1}^{\infty} x_l^{(\alpha)} z^l$ and $(\sum_{i}A_{i}D^{i})^{*}=\sum_{i}(-D)^{i}A_{i}^{\intercal}$, then we have the following bilinear equation for the $N$-component KP hierarchy
\begin{align}\label{N-bilinear}
 \operatorname{Res}_z \widetilde{\Psi}(\mathbf{x},z)\cdot\widetilde{\Psi}^{*}(\mathbf{x}',z)^{\intercal} = 0.
\end{align}
For the $N$-component KP hierarchy, there exist $N^{2}-N+1$ tau functions $\tau_{\alpha\beta}(\mathbf{x})\,(1\leq \alpha,\beta\leq N)$ with $\tau_{\alpha\alpha}(\mathbf{x})=\tau(\mathbf{x})$ such that \cite{kac2023,kac1993,zabrodin2019}
\begin{align}\label{wavematrix1}
  \widetilde{\Psi}_{\alpha\beta}(\mathbf{x}, z)=\varepsilon_{\alpha\beta} \frac{\tau_{\alpha\beta}\left(\mathbf{x}-[z^{-1}]_\beta\right)}{\tau(\mathbf{x})} z^{\delta_{\alpha\beta}-1} e^{\xi(\mathbf{x}^{(\beta)}, z)}, \quad
  \widetilde{\Psi}_{\beta\alpha}^{*}(\mathbf{x},z)=\varepsilon_{\beta\alpha}\frac{\tau_{\alpha\beta}\left(\mathbf{x}+[z^{-1}]_\alpha\right)}{\tau(\mathbf{x})} z^{\delta_{\beta\alpha}-1} e^{-\xi(\mathbf{x}^{(\alpha)}, z)},
\end{align}
where $\mathbf{x}\pm[z^{-1}]_{\beta}=(\mathbf{x}^{(1)},\cdots,\mathbf{x}^{(\beta-1)},\mathbf{x}^{(\beta)}\pm[z^{-1}],\mathbf{x}^{(\beta+1)},\cdots,\mathbf{x}^{(N)})$, $[z^{-1}]=(z^{-1}, z^{-2}/2, z^{-3}/3, \cdots)$, $\varepsilon_{\alpha\beta}=1$ if $\alpha\leq\beta$ and $\varepsilon_{\alpha\beta}=-1$ if $\alpha>\beta$.
Therefore \eqref{N-bilinear} is equivalent to the following bilinear equation \cite{kac2023,teo2011,zabrodin2019}:
\begin{align*}
  \sum_{\gamma=1}^{N} \operatorname{Res}_z \varepsilon_{\alpha \gamma} \varepsilon_{\beta \gamma} z^{\delta_{\alpha \gamma} +\delta_{ \beta \gamma} - 2}
e^{\xi(\mathbf{x}^{(\gamma)}- \mathbf{x}'^{(\gamma)}, z)}\tau_{\alpha \gamma}\bigl(\mathbf{x} - [z^{-1}]_\gamma\bigr)\tau_{\gamma\beta}\bigl(\mathbf{x}' + [z^{-1}]_\gamma\bigr) = 0,\quad 1\leq \alpha,\beta\leq N.
\end{align*}
\subsection{Main results}

To solve the constrained matrix KP hierarchy \eqref{matrixKPLaxeq}
\eqref{con-matrix-KP-Lax-operator} \eqref{qrderivative}, we firstly need to consider
the following reduction of the $N$-component KP hierarchy \cite{zhang1999}:
\begin{align}
  &\mathcal{L}(\mathbf{x},D)^{k}=\mathcal{L}(\mathbf{x},D)^{k}_{\geq0}+\sum_{j=1}^{m}\mathcal{Q}_{j}(\mathbf{x})\partial^{-1}\mathcal{R}_{j}(\mathbf{x})^{\intercal},\label{constrainlaxoperator2}\\
  &\mathcal{Q}_{i,x_{n}^{(\alpha)}}(\mathbf{x})=B_{n}^{(\alpha)}(\mathbf{x},D)(\mathcal{Q}_{i}(\mathbf{x})),\quad \mathcal{R}_{i,x_{n}^{(\alpha)}}(\mathbf{x})=-B_{n}^{(\alpha)}(\mathbf{x},D)^{*}(\mathcal{R}_{i}(\mathbf{x})),\label{qandrderivative2}
\end{align}
where $\mathcal{Q}_{i}(\mathbf{x})=(\mathcal{Q}_{i,\alpha\beta}(\mathbf{x}))_{1\leq\alpha,\beta\leq
N},\,\mathcal{R}_{i}(\mathbf{x})=(\mathcal{R}_{i,\alpha\beta}(\mathbf{x}))_{1\leq\alpha,\beta\leq
N}$ are $N\times N$ matrix functions. Here, we refer to
\eqref{LandC} \eqref{constrainlaxoperator2} \eqref{qandrderivative2} as the $(k,m)$-constrained
$N$-component KP hierarchy. We next attempt to express the constrained
$N$-component KP hierarchy in the form of bilinear equations with the help of tau functions, so that we can obtain the solutions of the constrained matrix KP hierarchy by multi-component boson-fermion correspondence\cite{zabrodin2019,kac2023,kac1993}.
\begin{theorem}\label{theorembilinear}
Given the $(k,m)$-constrained $N$-component KP hierarchy defined by
\eqref{LandC} \eqref{constrainlaxoperator2} \eqref{qandrderivative2}, together with the corresponding
matrix wave function $\widetilde{\Psi}(\mathbf{x},z)$ and matrix adjoint
wave function $\widetilde{\Psi}^{*}(\mathbf{x},z)$, the following bilinear
equations hold
\begin{align}
&\operatorname{Res}_{z}z^{k}\widetilde{\Psi}(\mathbf{x},z)\widetilde{\Psi}^*(\mathbf{x}',z)^\intercal=\sum_{i=1}^{m}\mathcal{Q}_{i}(\mathbf{x})\mathcal{R}_{i}(\mathbf{x}')^\intercal ,\label{qandrbilinearequation1}\\
&\operatorname{Res}_z\widetilde{\Psi}(\mathbf{x},z)\cdot\Omega(\mathcal{Q}_{i}(\mathbf{x}'), \widetilde{\Psi}^*(\mathbf{x}',z))=-\mathcal{Q}_{i}(\mathbf{x}), \label{qbilinearequation1}\\
&\operatorname{Res}_{z} \Omega\bigl(\widetilde{\Psi}(\mathbf{x},z), \mathcal{R}_{i}(\mathbf{x})\bigr)\cdot\widetilde{\Psi}^*(\mathbf{x}',z)^\intercal=\mathcal{R}_{i}(\mathbf{x}')^\intercal ,\label{rbilinearequation1}
\end{align}
where $\Omega(F, G)$ is determined by $\Omega(F, G)_{x_{n}^{(\alpha)}}=\operatorname{Res}_{D}(D^{-1}G^{\intercal}B_{n}^{(\alpha)}FD^{-1})$ for two $N\times N$ matrix functions $F$ and $G$ of $\mathbf{x}$ (see\cite{oevel1993}), and $\operatorname{Res}_{D}\sum_{i}A_{i}D^{i}=A_{-1}$.

Conversely, assume that the matrix functions
$\widetilde{\Psi}(\mathbf{x},z), \widetilde{\Psi}^{*}(\mathbf{x},z),
\mathcal{Q}_{i}(\mathbf{x})$ and  $\mathcal{R}_{i}(\mathbf{x})$ satisfy
\eqref{qandrbilinearequation1}-\eqref{rbilinearequation1}, where $\widetilde{\Psi}(\mathbf{x},z)$ and
$\widetilde{\Psi}^{*}(\mathbf{x},z)$ take the following special forms:
$$\widetilde{\Psi}(\mathbf{x},z)=(I_{N}+\sum_{i=1}^{\infty}W_{i}(\mathbf{x})z^{-i})e^{\xi(\mathbf{x},z)}, \quad \widetilde{\Psi}^{*}(\mathbf{x},z)=(I_{N}+\sum_{i=1}^{\infty}V_{i}(\mathbf{x})z^{-i})e^{-\xi(\mathbf{x},z)}.$$
If we set $\mathcal{W}(\mathbf{x})=I_{N}+\sum_{i=1}^{\infty}W_{i}(\mathbf{x})D^{i}$ and $
\mathcal{L}(\mathbf{x}, D)=\mathcal{W}(\mathbf{x}, D)D\mathcal{W}(\mathbf{x}, D)^{-1}, \,\mathcal{C}^{(\alpha)}= \mathcal{W}(\mathbf{x},D) E_{\alpha\alpha} \mathcal{W}(\mathbf{x},D)^{-1}$, then $\mathcal{L}(\mathbf{x}, D),$
$\mathcal{Q}_{i}(\mathbf{x})$ and $\mathcal{R}_{i}(\mathbf{x})\,(1\leq i\leq m)$ satisfy the
$(k,m)$-constrained $N$-component KP hierarchy
\eqref{LandC} \eqref{constrainlaxoperator2} \eqref{qandrderivative2}.
\end{theorem}
\begin{theorem}\label{theoremtau}
Assume that $\tau_{\alpha\beta}(\mathbf{x})$ with $\tau_{\alpha\alpha}(\mathbf{x})=\tau(\mathbf{x})$ are the tau functions of the $(k,m)$-constrained $N$-component KP hierarchy defined by \eqref{LandC} \eqref{constrainlaxoperator2} \eqref{qandrderivative2}, and let $\rho_{i,\alpha\beta} (\mathbf{x})= \mathcal{Q}_{i,\alpha\beta}(\mathbf{x})\tau(\mathbf{x}), \, \sigma_{i,\alpha\beta} (\mathbf{x}) = \mathcal{R}_{i,\alpha\beta}(\mathbf{x})\tau(\mathbf{x})$.  Then the following relations hold:
\begin{align}
&\sum_{\gamma=1}^N \operatorname{Res}_z\varepsilon_{\alpha \gamma}\varepsilon_{\beta\gamma}z^{k+\delta_{\alpha \gamma}+\delta_{\beta \gamma}-2}\tau_{\alpha \gamma}\bigl(\mathbf{x}-[z^{-1}]_{\gamma}\bigr) \tau_{\gamma\beta}\bigl(\mathbf{x}'+[z^{-1}]_{\gamma}\bigr)e^{\xi(\mathbf{x}^{(\gamma)}-\mathbf{x}'^{(\gamma)},z)}=
\sum_{\gamma=1}^N\sum_{i=1}^m\rho_{i,\alpha\gamma}(\mathbf{x})\sigma_{i,\beta\gamma}(\mathbf{x}'),\label{tau11}\\
&\sum_{\gamma=1}^N \operatorname{Res}_z\varepsilon_{\alpha \gamma} z^{\delta_{\alpha \gamma}-2} \tau_{\alpha \gamma}\bigl(\mathbf{x}-[z^{-1}]_{\gamma}\bigr) \rho_{i,\gamma\beta}\bigl(\mathbf{x}'+[z^{-1}]_\gamma\bigr) e^{\xi(\mathbf{x}^{(\gamma)}-\mathbf{x}'^{(\gamma)},z)}=
\tau(\mathbf{x}')\rho_{i,\alpha\beta}(\mathbf{x}),\label{tau2} \\
&\sum_{\gamma=1}^N \operatorname{Res}_z\varepsilon_{\beta\gamma} z^{\delta_{\beta \gamma}-2} \sigma_{i,\gamma\alpha}\bigl(\mathbf{x}-[z^{-1}]_\gamma\bigr) \tau_{\gamma\beta}\bigl(\mathbf{x}'+[z^{-1}]_{\gamma}\bigr)e^{\xi(\mathbf{x}^{(\gamma)}-\mathbf{x}'^{(\gamma)},z)}=
\sigma_{i,\beta\alpha}(\mathbf{x}') \tau(\mathbf{x}).\label{tau3}
\end{align}

Conversely, assume that $\tau_{\alpha\beta}(\mathbf{x})$ with $\tau_{\alpha\alpha}(\mathbf{x})=\tau(\mathbf{x})$, $\rho_{i,\alpha\beta} (\mathbf{x}),\, \sigma_{i,\alpha\beta} (\mathbf{x})$
satisfy \eqref{tau11}-\eqref{tau3} and set
\begin{align*}
&\mathcal{W}(\mathbf{x},z)_{\alpha\beta}=\varepsilon_{\alpha\beta} \frac{\tau_{\alpha\beta}\left(\mathbf{x}-[z^{-1}]_\beta\right)}{\tau(\mathbf{x})} z^{\delta_{\alpha\beta}-1},\,\mathcal{Q}_{i,\alpha\beta}(\mathbf{x})= \frac{\rho_{i,\alpha\beta}(\mathbf{x})}{\tau(\mathbf{x})}, \, \mathcal{R}_{i,\alpha\beta}(\mathbf{x}) = \frac{\sigma_{i,\alpha\beta} (\mathbf{x})}{\tau(\mathbf{x})},\\
&\mathcal{L}(\mathbf{x},D)=\mathcal{W}(\mathbf{x},D) D \mathcal{W}(\mathbf{x},D)^{-1}, \, C^{(\alpha)}(\mathbf{x},D)=\mathcal{W}(\mathbf{x},D)E_{\alpha\alpha}\mathcal{W}(\mathbf{x},D)^{-1},
\end{align*}
then $\mathcal{L}(\mathbf{x},D), \,C^{(\alpha)}(\mathbf{x},D), \,\mathcal{Q}_{i}(\mathbf{x}), \,\mathcal{R}_{i}(\mathbf{x})$ will satisfy the the $(k,m)$-constrained $N$-component KP hierarchy \eqref{LandC} \eqref{constrainlaxoperator2} \eqref{qandrderivative2}.
\end{theorem}
Based upon above results, now we can consider the solutions of the constrained matrix KP hierarchy by multi-component boson-fermion correspondence \cite{kac2023,kac1993,zabrodin2019}. Firstly, let us introduce the Clifford algebra $\mathcal{A}$ generated by $\mathbf{1}$ and the multi-component charged free fermions $\psi_{j}^{(\alpha)}$ and  $\psi_{j}^{\ast(\alpha)}(j\in\mathbb{Z}, 1\leq\alpha\leq N)$
satisfying
 \begin{align}\label{fermioncondition}
   \psi_{j}^{(\alpha)}\psi_{k}^{\ast(\beta)}+\psi_{k}^{\ast(\beta)}\psi_{j}^{(\alpha)}=\delta_{\alpha\beta}\delta_{jk},\quad
   \psi_{j}^{(\alpha)}\psi_{k}^{(\beta)}+\psi_{k}^{(\beta)}\psi_{j}^{(\alpha)}=\psi_{j}^{\ast(\alpha)}\psi_{k}^{\ast(\beta)}+\psi_{k}^{\ast(\beta)}\psi_{j}^{\ast(\alpha)}=0.
 \end{align}
The vacuums $|0\rangle$ and $\langle0|$ satisfy the conditions below
\begin{equation}\label{vacuumdefinition}
\begin{aligned}
  \psi_{j}^{(\alpha)}|0\rangle=0\ (j<0),\quad  \psi_{j}^{\ast(\alpha)}|0\rangle=0\ (j \geq 0), \\
  \langle 0|\psi_{j}^{(\alpha)}=0\ (j \geq 0),\quad \langle 0 | \psi_{j}^{*(\alpha)}=0\ (j<0).
\end{aligned}
\end{equation}
And for $\mathbf{p}=(p_{1},\cdots,p_{N})\in \mathbb{Z}^N$, let us define
$|\mathbf{p}\rangle=\Psi_{p_N}^{*(N)} \cdots \Psi_{p_2}^{*(2)} \Psi_{p_1}^{*(1)}|0\rangle$, $\langle\mathbf{p}|=\langle 0| \Psi_{p_1}^{(1)}\Psi_{p_2}^{(2)} \cdots \Psi_{p_N}^{(N)}$,
where $\Psi_{p}^{*(\alpha)}=\psi_{p_{-1}}^{(\alpha)}\ldots\psi_{0}^{(\alpha)}, \, \Psi_{p}^{(\alpha)}=\psi_{0}^{*(\alpha)}\ldots\psi_{p_{-1}}^{*(\alpha)}$ for $p>0$, while when $p<0$, $\Psi_{p}^{*(\alpha)}=\psi_{p}^{*(\alpha)}\ldots\psi_{-1}^{*(\alpha)}, \, \Psi_{p}^{(\alpha)}=\psi_{-1}^{(\alpha)}\ldots\psi_{p}^{(\alpha)}$ and $\Psi_{0}^{*(\alpha)}=\Psi_{0}^{(\alpha)}=1$.

 For one element $a\in\mathcal{A}$, we can define the vacuum expectation value $\langle 0|a|0\rangle$ by \eqref{fermioncondition} \eqref{vacuumdefinition} and $\langle0|0\rangle=1$, which can be computed by Wick theorem (see \eqref{wicktheorem} in Appendix). Here by the vacuum expectation value, we can realize the fermionic element $a|0\rangle$ into the bosonic element, which is just the boson-fermion correspondence.
For $\beta^{(i)}=\sum_{j \in \mathbb{Z}} a_{j}^{(i)} \psi_j^{(i)}\in V^{(i)}=\bigoplus_{j\in \mathds{Z}}\mathbb{C}\psi_{j}^{(i)},\, \beta^{*(i)}=\sum_{j \in \mathbb{Z}} b_{j}^{(i)} \psi_j^{*(i)}\in V^{*(i)}=\bigoplus_{j\in \mathbb{Z}}\mathbb{C}\psi_{j}^{*(i)}$, we can define
  \begin{align}\label{betadefination}
\beta_{[n]}^{(i)} = \sum_{j \in \mathbb{Z}} a_{j+n}^{(i)} \psi_j^{(i)},\quad\beta^{*{(i)}}_{[n]} = \sum_{j \in \mathbb{Z}} b_{j+n}^{(i)} \psi_j^{*(i)}.
\end{align}
If we further set $\beta = \sum_{i=1}^N \beta^{(i)}\in V=\bigoplus_{i=1}^{N}V^{(i)}$ and $\beta^* = \sum_{i=1}^N \beta^{*(i)}\in V^{*}=\bigoplus_{i=1}^{N}V^{*(i)}$, then we can also define $\beta_{[n]} = \sum_{i=1}^N \beta_{[n]}^{(i)}$ and $\beta^*_{[n]} = \sum_{i=1}^N \beta^{*{(i)}}_{[n]}$.

For the $(k,m)$-constrained $N$-component KP hierarchy, given $\alpha_{ai}$, $\Phi_{bi}\in V$ and $\alpha_{ci}^{*}$, $\Phi_{di}^{*}\in V^{*}$ for $1\leq a\leq P, \,1\leq b\leq K, \, 1\leq c\leq G, \,1\leq d\leq H$ and $1\leq i \leq N$, where $\Phi_{bi}$ and $\Phi_{dj}^{*}$ are required to satisfy $\Phi_{bi,[k]}\cdot\Phi_{bi}=\Phi_{dj,[-k]}^{*}\cdot\Phi_{dj}^{*}=0$, let us define
\begin{align*}
&\vec{\beta}_{a} \triangleq \bigl(\alpha_{a1},\,\alpha_{a1,[k]},\,\cdots,\,\alpha_{a1,[M_{a1}k]}, \,\cdots,\,\alpha_{aN},\,\alpha_{aN,[k]},\,\cdots,\,\alpha_{aN,[M_{aN}k]}\bigr),\quad
\vec{\Phi}_{b} \triangleq\bigl(\Phi_{b1},\,\cdots,\,\Phi_{bN}\bigr),\\
&\vec{\beta}^{*}_{c} \triangleq \bigl(\alpha^{*}_{c1},\, \alpha^{*}_{c1,[-k]},\,\cdots,\,\alpha^{*}_{c1,[-\widetilde{M}_{c1}k]},\,\cdots,\,\alpha^{*}_{cN},\,\alpha^{*}_{cN,[-k]},\,\cdots,\,\alpha^{*}_{cN,[-\widetilde{M}_{cN}k]}\bigr),\quad
\vec{\Phi}^{*}_{d} \triangleq\bigl(\Phi^{*}_{d1},\,\cdots,\,\Phi^{*}_{dN}\bigr),
\end{align*}
and denote
\begin{align*}
(\beta_{1},\cdots,\beta_{\mathcal{M}}) = \bigl(\vec{\beta}_{1},\cdots,\vec{\beta}_{P},\vec{\Phi}_{1},\cdots,\vec{\Phi}_{K}\bigr),\quad (\beta_{1}^{*},\cdots, \beta_{\mathcal{N}}^{*}) = \bigl(\vec{\beta}^{*}_{1},\cdots,\vec{\beta}^{*}_{G},\vec{\Phi}^{*}_{1},\cdots,\vec{\Phi}^{*}_{H}\bigr),
\end{align*}
where $m=P+G, \, \mathcal{M} = \sum_{a=1}^{P}\sum_{i=1}^{N}M_{ai} + (P+K)N, \, \mathcal{N} = \sum_{c=1}^{G}\sum_{j=1}^{N}\widetilde{M}_{cj} + (G+H)N$.
For convenience, let us further introduce for $1 \leq a\leq P, \, 1 \leq c\leq G$ and $1 \leq \gamma\leq N$,
\begin{equation}\label{ABCD}
\begin{aligned}
&\tau^{[\overrightarrow{\mathcal{M}},\overrightarrow{\mathcal{N}}]}=\beta ^{*[\overrightarrow{\mathcal{N}}]}\beta ^{[\overrightarrow{\mathcal{M}}]},\quad A_{a\gamma}^{[\overrightarrow{\mathcal{M}},\overrightarrow{\mathcal{N}}]}=\beta ^{*[\overrightarrow{\mathcal{N}}]}\beta ^{[\overrightarrow{\mathcal{M}}\setminus{\overrightarrow{p(a,\gamma)}}]}\cdot\beta _{a\gamma,[(M_{a\gamma}+1)k]}\beta ^{[\overrightarrow{p(a,\gamma)}]},\quad
B_{a\gamma}^{[\overrightarrow{\mathcal{M}},\overrightarrow{\mathcal{N}}]}=\beta ^{*[\overrightarrow{\mathcal{N}}]}\beta ^{[\overrightarrow{\mathcal{M}}\setminus\{p(a,\gamma)\}]},\\
&C_{c\gamma}^{[\overrightarrow{\mathcal{M}},\overrightarrow{\mathcal{N}}]}=\beta ^{*[\overrightarrow{\mathcal{N}}\setminus\{\tilde{p}(c,\gamma)\}]}\beta ^{[\overrightarrow{\mathcal{M}}]},\quad
D_{c\gamma}^{[\overrightarrow{\mathcal{M}},\overrightarrow{\mathcal{N}}]}=\beta ^{*[\overrightarrow{\mathcal{N}}\setminus\overrightarrow{\tilde{p}(c,\gamma)}]}\beta _{c\gamma,[-(\tilde{M}_{c\gamma}+1)k]}^*\beta ^{*[\overrightarrow{\tilde{p}(c,\gamma)}]}\beta ^{[\overrightarrow{\mathcal{M}}]},
\end{aligned}
\end{equation}
where we have assumed
\begin{align}\label{signbb*}
\beta^{[\overrightarrow{\mathcal{M}}]}=\beta_\mathcal{M} \cdots \beta_1,\quad \beta^{*[\overrightarrow{\mathcal{M}}]}=\beta_\mathcal{M}^* \cdots \beta_1^*,
\end{align}
and
$p(a,\gamma) = \sum_{b=1}^{a-1}\sum_{j=1}^N \mathcal{M}_{b j} + N(a-1)+\sum_{l=1}^{\gamma}M_{al}+\gamma, \,
\tilde{p}(c,\gamma) = \sum_{b=1}^{c-1}\sum_{l=1}^N \widetilde{\mathcal{M}}_{b l} + N(c-1)+\sum_{l=1}^{\gamma}\widetilde{M}_{cl}+\gamma $, $\overrightarrow{\mathcal{N}}=(\mathcal{N}, \mathcal{N}-1,\ldots, 2, 1)$, $\overrightarrow{\mathcal{N}}\setminus\{i\}=\{ \mathcal{N}, \mathcal{N}-1, \ldots, i+1, i-1, \ldots, 1 \}$, $\overrightarrow{\mathcal{N}}\setminus\overrightarrow{i}=\{ \mathcal{N}, \mathcal{N}-1, \ldots, i+1 \}$. Then we have  the following theorem.
\begin{theorem}\label{theorem:solution}
For $\tau^{[\overrightarrow{\mathcal{M}},\overrightarrow{\mathcal{N}}]}$, $A_{i\gamma}^{[\overrightarrow{\mathcal{M}},\overrightarrow{\mathcal{N}}]}$, $B_{i\gamma}^{[\overrightarrow{\mathcal{M}},\overrightarrow{\mathcal{N}}]}$, $C_{i\gamma}^{[\overrightarrow{\mathcal{M}},\overrightarrow{\mathcal{N}}]}$ and $D_{i\gamma}^{[\overrightarrow{\mathcal{M}},\overrightarrow{\mathcal{N}}]}$ given in \eqref{ABCD}, if we set
\begin{align}
&\tau_{\alpha\beta}(\mathbf{x})=\langle\mathbf{p}+\mathbf{e}_\alpha - \mathbf{e}_\beta | e^{H(\mathbf{x})} \tau^{[\overrightarrow{\mathcal{M}},\overrightarrow{\mathcal{N}}]} |0\rangle,\label{solution1}\\
&\rho_{i,\alpha \gamma}(\mathbf{x})
=\left\{
\begin{array}{ll}
\langle\mathbf{p}+\mathbf{e}_\alpha \left|\right. e^{H(\mathbf{x})}A_{i\gamma}^{[\overrightarrow{\mathcal{M}},\overrightarrow{\mathcal{N}}]} |0\rangle, & \hbox{$1 \leq i \leq P$;} \\
\langle\mathbf{p}+\mathbf{e}_\alpha \left|\right. e^{H(\mathbf{x})}C_{i-P, \gamma}^{[\overrightarrow{\mathcal{M}},\overrightarrow{\mathcal{N}}]} |0\rangle, & \hbox{$P+1 \leq i \leq m$,}
\end{array}
\right.\label{solution2}\\
&\sigma_{i,\beta \gamma}(\mathbf{x})=
\left\{
\begin{array}{ll}
\langle\mathbf{p}-\mathbf{e}_\beta \left|\right. e^{H(\mathbf{x})}B_{i\gamma}^{[\overrightarrow{\mathcal{M}},\overrightarrow{\mathcal{N}}]} |0\rangle \, & \hbox{$1 \leq i \leq P$;} \\
\langle\mathbf{p}-\mathbf{e}_\beta \left|\right. e^{H(\mathbf{x})}D_{i-P, \gamma}^{[\overrightarrow{\mathcal{M}},\overrightarrow{\mathcal{N}}]} |0\rangle, & \hbox{$P+1 \leq i \leq m$,}
\end{array}
\right.\label{solution3}
\end{align}
where $\mathbf{p}\in \mathds{Z}^{N}$ is chosen such that $\langle\mathbf{p}| e^{H(\mathbf{x})} \tau^{[\overrightarrow{\mathcal{M}},\overrightarrow{\mathcal{N}}]} |0\rangle\neq 0$,  $H(\mathbf{x})= \sum_{\alpha=1}^N \sum_{n=1}^\infty x^{(\alpha)}_{n}\sum_{j \in \mathbb{Z}}  \psi_j^{(\alpha)} \psi_{j+n}^{\ast(\alpha)}$ and $\mathbf{e}_\alpha\in \mathbb{Z}^N$ (with the unique nonzero element $1$ at the \(\alpha\)-th position), then $\tau_{\alpha\beta}(\mathbf{x})$, $\rho_{i,\alpha
\gamma}(\mathbf{x})$ and $\sigma_{i,\beta \gamma}(\mathbf{x})$ satisfy the
bilinear equations \eqref{tau11}-\eqref{tau3} of the $(k,m)$-constrained $N$-component KP hierarchy.

In particular, if we denote
$\mathcal{L}(\mathbf{x},D)=\mathcal{W}(\mathbf{x},D) D \mathcal{W}(\mathbf{x},D)^{-1}$, $\mathcal{Q}_{i,\alpha\beta}(\mathbf{x})= \frac{\rho_{i,\alpha\beta} (\mathbf{x})}{\tau(\mathbf{x})}$, $\mathcal{R}_{i,\alpha\beta}(\mathbf{x})= \frac{\sigma_{i,\alpha\beta} (\mathbf{x})}{\tau(\mathbf{x})}$,
where $\mathcal{W}(\mathbf{x},z)_{\alpha\beta}=\varepsilon_{\alpha\beta} \frac{\tau_{\alpha\beta}\left(\mathbf{x}-[z^{-1}]_\beta\right)}{\tau(\mathbf{x})} z^{\delta_{\alpha\beta}-1}$, then $$L(\mathbf{t},\partial)=\mathcal{L}(\mathbf{x},D)|_{\mathbf{x}^{(\alpha)}=\mathbf{t}, 1\leq\alpha\leq N},\quad Q_i(\mathbf{t})=\mathcal{Q}_{i}(\mathbf{x})|_{\mathbf{x}^{(\alpha)}=\mathbf{t}, 1\leq\alpha\leq N},\quad R_i(\mathbf{t})=\mathcal{R}_{i}(\mathbf{x})|_{\mathbf{x}^{(\alpha)}=\mathbf{t}, 1\leq\alpha\leq N},$$
will satisfy the $(k,m)$-constrained matrix KP hierarchy \eqref{matrixKPLaxeq}\eqref{con-matrix-KP-Lax-operator}\eqref{qrderivative}.
\end{theorem}

\section{Bilinear equations of the constrained multi-component KP hierarchy}\label{section2}
In this section, we will consider the bilinear equations of the $(k,m)$-constrained $N$-component KP hierarchy, in the forms of wave functions and tau functions. The corresponding results are the key to construct the tau functions by the multi-component boson-fermion correspondence in what follows.

\subsection{Bilinear equations in the form of wave functions}

\begin{lemma}\label{PandQ*}{\cite{kac2023,kac1993}}
For any two matrix pseudo-differential operators $P(\mathbf{x},D)$ and $Q(\mathbf{x},D)$, the following identity holds:
 $$\operatorname{Res}_{z} P(\mathbf{x},D) \left( \sum_{\alpha=1}^N E_{\alpha\alpha}  e^{\xi(\mathbf{x}^{(\alpha)},z)} \right) \cdot \left( Q(\mathbf{x},D) \left( \sum_{\beta=1}^N E_{\beta\beta} e^{-\xi(\mathbf{x}^{(\beta)},z)} \right) \right)^\intercal
= \operatorname{Res}_{D}P(\mathbf{x},D) Q^{*}(\mathbf{x},D).$$
And if the following relation holds
  $$\mathrm{Res}_z P(\mathbf{x},D) \left( \sum_{\alpha=1}^N E_{\alpha\alpha} e^{\xi(\mathbf{x}^{(\alpha)},z)} \right)\cdot \left( Q(\mathbf{x'},D) \left( \sum_{\beta=1}^N E_{\beta\beta} e^{-\xi(\mathbf{x'}^{(\beta)},z)} \right) \right)^\intercal = 0,$$
then $(P(\mathbf{x},D) \cdot Q^{*}(\mathbf{x},D))_{<0} = 0.$
\end{lemma}
\begin{lemma}\label{lemmaoperator}\cite{oevel1993}
For any matrix pseudo-differential operator $A(\mathbf{x},D)$ and any matrix function $f(\mathbf{x})$, we have
$$(A(\mathbf{x},D)_{\geq0}f(\mathbf{x})D^{-1})_{<0}=A(\mathbf{x},D)_{\geq0}(f(\mathbf{x}))D^{-1}, \quad (D^{-1}f(\mathbf{x})A(\mathbf{x},D)_{\geq0})_{<0}=D^{-1}(A(\mathbf{x},D)^{*}_{\geq0}(f(\mathbf{x})^{\intercal}))^{\intercal}.$$
\end{lemma}
\begin{lemma}\label{lemmaeigenfunction}
Given the eigenfunction $\mathcal{Q}(\mathbf{x})$ and the adjoint eigenfunction $\mathcal{R}(\mathbf{x})$ of the $N$-component KP hierarchy satisfying $\mathcal{Q}(\mathbf{x})_{x_n^{(\alpha)}}
= B_n^{(\alpha)}(\mathbf{x},D) \left( \mathcal{Q}(\mathbf{x}) \right),\,
\mathcal{R}(\mathbf{x})_{x_n^{(\alpha)}} = - B_n^{(\alpha)}(\mathbf{x},D)^{*} \left( \mathcal{R}(\mathbf{x}) \right)$, there exist matrix differential operators $\mathcal{A}_n^{(\alpha)}(\mathbf{x},D)$ and $\mathcal{F}_n^{(\alpha)}(\mathbf{x},D)$ such that
$$\Omega\bigl(\mathcal{Q}(\mathbf{x}), \mathcal{R}(\mathbf{x})\bigr)_{x_{n}^{(\alpha)}}=\mathcal{A}_n^{(\alpha)}(\mathbf{x},D)(\mathcal{Q}(\mathbf{x}))=\Big(\mathcal{F}_n^{(\alpha)}(\mathbf{x},D)(\mathcal{R}(\mathbf{x}))\Big)^{\intercal}.$$
\end{lemma}
\begin{proof}
 Firstly by the definition of $\Omega$, we have
 \begin{align*}
\Omega(\mathcal{Q}, \mathcal{R})_{x_{n}^{(\alpha)}}= \operatorname{Res}_{D} \left( D^{-1} \mathcal{R}^{\intercal} B_{n}^{(\alpha)} \mathcal{Q} D^{-1} \right)=\operatorname{Res}_{D} \left(D^{-1}<\mathcal{R}^{\intercal} B_{n}^{(\alpha)} \mathcal{Q}>\cdot D^{-1} - D^{-1} \cdot D^{-1}<\mathcal{R}^{\intercal} B_{n}^{(\alpha)} \mathcal{Q}> \right).
\end{align*}
Here the symbol $<>$ is defined by: $A<B>=\sum_{i}A(b_{i})D^{i}$
for $B=\sum_{i}b_{i}D^{i}$. In particular, $D(\Omega(\mathcal{Q}, \mathcal{R}))=\mathcal{R}^{\intercal}\mathcal{Q}$. Furthermore according to Lemma
\ref{lemmaoperator}
\begin{align*}
\left(D^{-1}<\mathcal{R}^{\intercal} B_{n}^{(\alpha)} \mathcal{Q}> D^{-1} \right)_{<0}&= D^{-1}<\left(\mathcal{R}^{\intercal} B_{n}^{(\alpha)} \mathcal{Q}D^{-1}\right)_{<0}>=D^{-1}\left(\mathcal{R}^{\intercal} \cdot B_{n}^{(\alpha)}(\mathcal{Q})\right)\cdot D^{-1}\\
&=D^{-1}\left(B_{n}^{(\alpha)}(Q)^{\intercal}\cdot\mathcal{R}\right)^{\intercal}\cdot D^{-1}.\\
\Big(D^{-1} \cdot D^{-1} < \mathcal{R}^{\intercal} B_{n}^{(\alpha)} \mathcal{Q} > \Big)_{<0}
&= D^{-1} < \Big(D^{-1}\mathcal{R}^{\intercal} B_{n}^{(\alpha)} \mathcal{Q}\Big)_{<0} >
= D^{-1} < D^{-1} \cdot \Big(\big(B_{n}^{(\alpha)} \mathcal{Q}\big)^{*}(\mathcal{R})\Big)^{\intercal} > \\
&= D^{-1} \cdot D^{-1}  \Big((B_{n}^{(\alpha)*}(\mathcal{R}))^{\intercal} \cdot \mathcal{Q} \Big)
= D^{-1} \cdot D^{-1} \Big( \mathcal{Q}^{\intercal}\cdot B_{n}^{(\alpha)*}(\mathcal{R}) \Big)^{\intercal}.
\end{align*}
Therefore $\Omega\bigl(\mathcal{Q}(\mathbf{x}), \mathcal{R}(\mathbf{x})\bigr)_{x_{n}^{(\alpha)}}=\mathcal{A}_n^{(\alpha)}(\mathbf{x},D)(\mathcal{Q}(\mathbf{x}))=(\mathcal{F}_n^{(\alpha)}(\mathbf{x},D)(\mathcal{R}(\mathbf{x})))^{\intercal},$
where
$$\mathcal{A}_n^{(\alpha)}(\mathbf{x},D)=D^{-1} \cdot
\mathcal{R}(\mathbf{x})^{\intercal}\cdot B_{n}^{(\alpha)}(\mathbf{x},D)-D^{-1} \cdot\left(
B_{n}^{(\alpha)}(\mathbf{x},D)^{*} (\mathcal{R}(\mathbf{x})) \right)^{\intercal},$$ $$\mathcal{F}_n^{(\alpha)}(\mathbf{x},D)=D^{-1} \cdot
B_{n}^{(\alpha)}(\mathbf{x},D)(\mathcal{Q}(\mathbf{x}))^{\intercal}-D^{-1} \cdot
\mathcal{Q}(\mathbf{x})^{\intercal}\cdot B_{n}^{(\alpha)}(\mathbf{x},D)^{*}.$$
By Lemma \ref{lemmaoperator},
we can find
$(\mathcal{A}_n^{(\alpha)}(\mathbf{x},D))_{<0}=(\mathcal{F}_n^{(\alpha)}(\mathbf{x},D))_{<0}=0,$ which
implies that both $\mathcal{A}_n^{(\alpha)}(\mathbf{x},D)$ and
$\mathcal{F}_n^{(\alpha)}(\mathbf{x},D)$ are differential operators.
\end{proof}
\begin{proposition}\label{propositionqandr}
Given the eigenfunction $\mathcal{Q}(\mathbf{x})$ and the adjoint eigenfunction $\mathcal{R}(\mathbf{x})$ of the $N$-component KP hierarchy, together with the corresponding wave function $\widetilde{\Psi}(\mathbf{x},z)$ and adjoint wave function $\widetilde{\Psi}^{*}(\mathbf{x},z)$, the following bilinear identities hold:
\begin{align}
&\operatorname{Res}_z\widetilde{\Psi}(\mathbf{x}',z)\cdot\Omega(\mathcal{Q}(\mathbf{x}), \widetilde{\Psi}^*(\mathbf{x},z))=-\mathcal{Q}(\mathbf{x}'),
\label{qbilinearequation2}\\
&\operatorname{Res}_{z} \Omega\bigl(\widetilde{\Psi}(\mathbf{x},z), \mathcal{R}(\mathbf{x})\bigr)\cdot \widetilde{\Psi}^*(\mathbf{x}',z)^\intercal=\mathcal{R}(\mathbf{x}')^\intercal .\label{rbilinearequation2}
\end{align}
\end{proposition}
\begin{proof}
First, set the left-hand side of \eqref{qbilinearequation2} as
$I(\mathbf{x}, \mathbf{x}'),$ then by Lemma \ref{lemmaeigenfunction},
$$I(\mathbf{x}, \mathbf{x}')_{x_{n}^{(\alpha)}}=\operatorname{Res}_z\widetilde{\Psi}(\mathbf{x}',z)\cdot\Omega(\mathcal{Q}(\mathbf{x}), \widetilde{\Psi}^*(\mathbf{x},z))_{x_{n}^{(\alpha)}}=\operatorname{Res}_z\widetilde{\Psi}(\mathbf{x}',z)\Bigl(\mathcal{F}_n^{(\alpha)}(\mathbf{x},D)(\widetilde{\Psi}^{*}(\mathbf{x},z))\Bigr)^{\intercal}.$$
Together with $\operatorname{Res}_{z}\widetilde{\Psi}(\mathbf{x},z)\widetilde{\Psi}^*(\mathbf{x}',z)^\intercal = 0,$ we obtain $I(\mathbf{x}, \mathbf{x}')_{x_{n}^{(\alpha)}}=0,$ which means that $I(\mathbf{x}, \mathbf{x}')$ only depends on $\mathbf{x}'.$ So we can assume $I(\mathbf{x}, \mathbf{x}')=J(\mathbf{x}').$ Therefore, it can be observed that $J(\mathbf{x})=\operatorname{Res}_z\widetilde{\Psi}(\mathbf{x},z)\Omega(\widetilde{\Psi}^*(\mathbf{x},z), \mathcal{Q}(\mathbf{x}))^{\intercal}.$ By Lemma \ref{PandQ*}, we have
\begin{align*}
J(\mathbf{x})=-\operatorname{Res}_D\mathcal{W}(\mathbf{x},D)\mathcal{W}^{-1}(\mathbf{x},D)\mathcal{Q}(\mathbf{x})D^{-1}
=-\mathcal{Q}(\mathbf{x}).
\end{align*}
So \eqref{qbilinearequation2} is proved, and \eqref{rbilinearequation2} can
be proved in a similar way.
\end{proof}
With the above preparation, we are now in a position to prove {\bf Theorem
	\ref{theorembilinear}}. First, we prove the first part of Theorem
\ref{theorembilinear}: \eqref{LandC} \eqref{constrainlaxoperator2} \eqref{qandrderivative2} $\Rightarrow$
\eqref{qandrbilinearequation1}-\eqref{rbilinearequation1}. Notice that the
proofs of \eqref{qbilinearequation1} \eqref{rbilinearequation1} are
contained in Proposition \ref{propositionqandr}. Thus, it suffices to prove \eqref{qandrbilinearequation1}. In fact, according to
\eqref{matrixLaxdefination} and
$$\mathcal{L}(\mathbf{x},D)^{n}(\Psi(\mathbf{x},z))=z^{n}\Psi(\mathbf{x},z),\, \Psi(\mathbf{x},z)_{x_{n}^{(\alpha)}}=B_{n}^{(\alpha)}(\mathbf{x},D)(\Psi(\mathbf{x},z)),$$ we have
\begin{align*}
\widetilde{\Psi}(\mathbf{x},z)_{x_{k}^{(\alpha)}}=z^{k}\widetilde{\Psi}(\mathbf{x},z)-\sum_{i=1}^{m}\mathcal{Q}_{i}(\mathbf{x})\Omega\bigl(\widetilde{\Psi}(\mathbf{x},z), \mathcal{R}_{i}(\mathbf{x})\bigr).
\end{align*}
Then by
$\operatorname{Res}_z\widetilde{\Psi}(\mathbf{x},z)_{x_{k}^{(\alpha)}}\cdot\widetilde{\Psi}^{*}(\mathbf{x}',z)^{\intercal}
= 0$ derived from \eqref{N-bilinear}, we can obtain
\eqref{qandrbilinearequation1} with the help of \eqref{rbilinearequation1}.

We now prove the second part of Theorem \ref{theorembilinear}, namely
\eqref{qandrbilinearequation1}-\eqref{rbilinearequation1} $\Rightarrow$
\eqref{LandC} \eqref{constrainlaxoperator2} \eqref{qandrderivative2}. Applying
$D'=\sum_{\alpha=1}^{N}\partial_{x_1^{(\alpha)'} }$ to both sides of
\eqref{qbilinearequation1}, and using $D'(\Omega(\mathcal{Q}_{i}(\mathbf{x}')), \widetilde{\Psi}^*(\mathbf{x}',z))=\widetilde{\Psi}^*(\mathbf{x}',z)^{\intercal}\mathcal{Q}_{i}(\mathbf{x}')$, we have
\begin{align*}
 \operatorname{Res}_z \widetilde{\Psi}(\mathbf{x},z)\cdot\widetilde{\Psi}^{*}(\mathbf{x}',z)^{\intercal} = 0.
\end{align*}
Noting that $\widetilde{\Psi}(\mathbf{x},z) =\mathcal{W}(\mathbf{x},D)\left( \sum_{\alpha=1}^N E_{\alpha\alpha} e^{\xi(\mathbf{x}^{(\alpha)},z)} \right)$ and referring to the proof in \cite{kac2023,kac1993}, we can obtain
\begin{align}
&\mathcal{W}_{x_{n}^{(\alpha)}}(\mathbf{x},D)=-\left( \mathcal{W}(\mathbf{x},D) E_{\alpha\alpha} D^n \mathcal{W}(\mathbf{x},D)^{-1} \right)_{<0} \mathcal{W},\label{Wderivative1}\\
&\widetilde{\Psi}^{*}(\mathbf{x},z) ={{(\mathcal{W}(\mathbf{x},D)^{-1})^{*}}}\left( \sum_{\alpha=1}^N E_{\alpha\alpha} e^{-\xi(\mathbf{x}^{(\alpha)},z)} \right).\label{Wderivative2}
\end{align}
So if we set $\mathcal{L}(\mathbf{x},D) = \mathcal{W}(\mathbf{x},D) \partial \mathcal{W}(\mathbf{x},D)^{-1}, C^{(\alpha)}(\mathbf{x},D)
= \mathcal{W}(\mathbf{x},D) E_{\alpha\alpha} \mathcal{W}(\mathbf{x},D)^{-1}$, then $\mathcal{L}(\mathbf{x},D)$ and
$\mathcal{C}^{(\alpha)}(\mathbf{x},D)$ satisfy \eqref{LandC}.

By \eqref{Wderivative1} \eqref{Wderivative2} and $\widetilde{\Psi}(\mathbf{x},z)
=\mathcal{W}(\mathbf{x},D)\left( \sum_{\alpha=1}^N E_{\alpha\alpha}
e^{\xi(\mathbf{x}^{(\alpha)},z)} \right),$ we have
$$\widetilde{\Psi}(\mathbf{x},z)_{x_n^{(\alpha)}} =
B_n^{(\alpha)}(\mathbf{x},D)\bigl(\widetilde{\Psi}(\mathbf{x},z) \bigr), \quad
\widetilde{\Psi}^{*}(\mathbf{x},z)_{x_n^{(\alpha)}}
=-B_n^{(\alpha)}(\mathbf{x},D)^{*}\bigl(\widetilde{\Psi}^{*}(\mathbf{x},z) \bigr).$$
So further by \eqref{qbilinearequation1} \eqref{rbilinearequation1}, we have
\begin{align*}
\mathcal{Q}_i(\mathbf{x})_{x_n^{(\alpha)}}
= B_n^{(\alpha)}(\mathbf{x},D) \left( \mathcal{Q}_i(\mathbf{x}) \right),\quad
\mathcal{R}_i(\mathbf{x})_{x_n^{(\alpha)}} = - B_n^{(\alpha)}(\mathbf{x},D)^{*} \left( \mathcal{R}_i(\mathbf{x}) \right),
\end{align*}
which completes the proof of \eqref{qandrderivative2}.

The key remaining step is to prove \eqref{constrainlaxoperator2}. Apply
$(D')^{j}=(\sum_{\alpha=1}^{N}\partial_{x_{1}}^{(\alpha)'})^{j}$ to \eqref{qandrbilinearequation1} and let
$\mathbf{x}'=\mathbf{x}$, we can obtain
\begin{align*}
\operatorname{Res}_z \, z^k \widetilde{\Psi}(\mathbf{x},z) \cdot  D^j \left(\widetilde{\Psi}^*(\mathbf{x},z) \right)^\intercal= \sum_{i=1}^m \mathcal{Q}_i(\mathbf{x})\cdot D^j(\mathcal{R}_i(\mathbf{x})^\intercal).
\end{align*}
Further by Lemma \ref{PandQ*}, we can obtain $\operatorname{Res}_{D} (\mathcal{W} D^k \mathcal{W}^{-1} \cdot D^j) = (-1)^j \sum_{i=1}^m \mathcal{Q}_i(\mathbf{x}) \cdot D^j(\mathcal{R}_i(\mathbf{x})^\intercal)$.
Therefore
\begin{align*}
(\mathcal{L}^{k})_{<0}= \sum_{j=0}^\infty \operatorname{Res}_{D} (\mathcal{W} D^k \mathcal{W}^{-1} \cdot D^j) D^{-j-1} = \sum_{j=0}^\infty (-1)^j \sum_{i=1}^m \mathcal{Q}_i\cdot D^j(\mathcal{R}_i^{\intercal})\cdot D^{-j-1}
= \sum_{i=1}^m \mathcal{Q}_i D^{-1} \mathcal{R}_i^\intercal.
\end{align*}
So far, we have completed the proof of {\bf Theorem \ref{theorembilinear}}.

\subsection{Bilinear equations in the form of tau functions}

\begin{proposition}
For the eigenfunction $\mathcal{Q}(\mathbf{x})$ and the adjoint eigenfunction $\mathcal{R}(\mathbf{x})$ of the $N$-component KP hierarchy, the following results hold:
\begin{align}
&\Omega(\mathcal{Q}(\mathbf{x}),\widetilde{\Psi}^*(\mathbf{x},z))_{\alpha\beta} = -z^{-1}\widetilde{\Psi}^*_{\alpha\alpha}(\mathbf{x},z)\mathcal{Q}\left( \mathbf{x} + [z^{-1}]_\alpha \right)_{\alpha\beta}, \label{Omegaandtau1}\\
&\Omega(\widetilde{\Psi}(\mathbf{x},z),\mathcal{R}(\mathbf{x}))_{\alpha\beta} =z^{-1}\widetilde{\Psi}_{\beta\beta}(\mathbf{x},z)\mathcal{R}\left( \mathbf{x} - [z^{-1}]_\beta \right)_{\beta\alpha}.\label{Omegaandtau2}
\end{align}
\end{proposition}
\begin{proof}
Notice that $\Omega\left(\mathcal{Q}(\mathbf{x}), \widetilde{\Psi}^*(\mathbf{x},z)\right) = -\widehat{\Phi}(\mathbf{x},z)\left(\sum_{\alpha=1}^{N} E_{\alpha\alpha} e^{-\xi(\mathbf{x}^{(\alpha)}, z)}\right),$
and $\widehat{\Phi}(\mathbf{x},z) = \mathcal{Q}(\mathbf{x})^\intercal z^{-1} + \mathcal{O}(z^{-2})$. Further if we set $\widehat{\Psi}(\mathbf{x},z) = \widetilde{\Psi}(\mathbf{x},z)\left(\sum_{\alpha=1}^{N} E_{\alpha\alpha} e^{-\xi(\mathbf{x}^{(\alpha)}, z)}\right),$ then $\widehat{\Psi}(\mathbf{x},z) = 1 + \mathcal{O}(z^{-1})$.

Since $\mathcal{Q}(\mathbf{x})$ is the eigenfunction, so by Proposition \ref{propositionqandr}, we have \eqref{qbilinearequation2}. So if we set $\mathbf{x}'\rightarrow \mathbf{x}, \,\mathbf{x}\rightarrow \mathbf{x}-[p^{-1}]_{\alpha}$ in \eqref{qbilinearequation2}, then the $(\alpha,\beta)$-position will give rise to
  \begin{align*}
\mathcal{Q}(\mathbf{x})_{\alpha\beta}&= \operatorname{Res}_z \widehat{\Psi}(\mathbf{x},z)_{\alpha\alpha} \widehat{\Phi}(\mathbf{x} - [p^{-1}]_\alpha, z)_{\beta\alpha} \frac{1}{1-z/p}+\sum_{\gamma \neq \alpha}^{N} \operatorname{Res}_z \widehat{\Psi}(\mathbf{x},z)_{\alpha\gamma} \widehat{\Phi}(\mathbf{x} - [p^{-1}]_\alpha, z)_{\beta\gamma}\\
&=\operatorname{Res}_z \widehat{\Psi}(\mathbf{x},z)_{\alpha\alpha} \widehat{\Phi}(\mathbf{x} - [p^{-1}]_\alpha, z)_{\beta\alpha} \frac{1}{1-z/p},
  \end{align*}
Then by the formula \cite{dickey2003}:
\begin{align}\label{formula}
\operatorname{Res}_z \sum_{i=1}^{+\infty} a_i(p) z^{-i}\frac{1}{1-z/p}=p\sum_{i=1}^{+\infty} a_i(p) p^{-i},
\end{align}
and $\widehat{\Psi}(\mathbf{x},z)_{\alpha\alpha}=\frac{\tau\left(\mathbf{x}-[z^{-1}]_\alpha\right)}{\tau(\mathbf{x})}$, we have
$$\widehat{\Phi}(\mathbf{x}, p)_{\beta\alpha} = \mathcal{Q}(\mathbf{x}+[p^{-1}]_\alpha)_{\alpha\beta} \frac{\tau(\mathbf{x} + [p^{-1}]_\alpha)}{\tau(\mathbf{x})} \cdot p^{-1},$$
which implies \eqref{Omegaandtau1}. And \eqref{Omegaandtau2}
can be proved in a similar way.
\end{proof}
\begin{lemma}\label{D}
Given $\tau_{\alpha\beta}(\mathbf{x})$ and $\rho_{\alpha\beta} (\mathbf{x})$ with $\tau_{\alpha\alpha}(\mathbf{x})=\tau(\mathbf{x})$ satisfing
\begin{align}\label{starone}
\operatorname{Res}_z \sum_{\gamma=1}^{N} \varepsilon_{\alpha\gamma}
z^{\delta_{\alpha\gamma}-2}\tau_{\alpha\gamma}\bigl(\mathbf{x}-[z^{-1}]_\gamma\bigr)\rho_{\gamma\beta}\bigl(\mathbf{x}'+[z^{-1}]_\gamma\bigr)
e^{\xi(\mathbf{x}^{(\gamma)}-\mathbf{x}^{(\gamma)'},z)} = \tau(\mathbf{x}')\rho_{\alpha\beta}(\mathbf{x}),
\end{align}
we have
\begin{equation}
\begin{aligned}\label{startwo}
D \tau(\mathbf{x})\cdot\rho_{\alpha \beta}\bigl(\mathbf{x} + [\lambda^{-1}]_{\alpha}\bigr) - \tau(\mathbf{x}) \cdot D \rho_{\alpha \beta}\bigl(\mathbf{x} + [\lambda^{-1}]_{\alpha}\bigr)&+ \lambda \Big(\tau(\mathbf{x}) \rho_{\alpha \beta}\bigl(\mathbf{x} + [\lambda^{-1}]_{\alpha}\bigr)-\tau\bigl(\mathbf{x} + [\lambda^{-1}]_{\alpha}\bigr) \rho_{\alpha \beta}(\mathbf{x})\Big) \\
&= \sum_{\gamma \neq \alpha} \varepsilon_{\alpha \gamma} \tau_{\alpha \gamma}\bigl(\mathbf{x} + [\lambda^{-1}]_{\alpha}\bigr) \rho_{\gamma \beta}(\mathbf{x}).
\end{aligned}
\end{equation}
\end{lemma}
\begin{proof}
First, apply $D = \sum_{\eta=1}^{N} x_1^{(\eta)}$ to both sides of \eqref{starone} and set $\mathbf{x}-\mathbf{x}'=[\lambda^{-1}]_\alpha$, then the corrosponding $(\alpha, \beta)$-position will become
\begin{align*}
&\sum_{\gamma\neq\alpha}\operatorname{Res}_z\varepsilon_{\alpha\gamma}z^{-2}\Bigl( D \tau_{\alpha\gamma}\bigl(\mathbf{x}-[z^{-1}]_\gamma\bigr)+ z\tau_{\alpha\gamma}\bigl(\mathbf{x}-[z^{-1}]_\gamma\bigr)
\Bigr) \cdot \rho_{\gamma\beta}\bigl(\mathbf{x}'+[z^{-1}]_\gamma\bigr)\\
&+\operatorname{Res}_zz^{-1}
\Bigl(D \tau\bigl(\mathbf{x}-[z^{-1}]_\alpha\bigr)+ z \tau\bigl(\mathbf{x}-[z^{-1}]_\alpha\bigr)\Bigr) \rho_{\alpha\beta}\bigl(\mathbf{x}'+[z^{-1}]_\alpha\bigr)(1-z/\lambda)^{-1} = \tau(\mathbf{x}')\cdot D \rho_{\alpha\beta}(\mathbf{x}).
\end{align*}
Notice that the left side is
$$\sum_{\gamma\neq\alpha}\varepsilon_{\alpha\gamma} \tau_{\alpha\gamma}(\mathbf{x})\rho_{\gamma\beta}(\mathbf{x}')+\Bigl(D \tau\bigl(\mathbf{x}-[\lambda^{-1}]_\alpha\bigr) + \lambda \tau\bigl(\mathbf{x}-[\lambda^{-1}]_\alpha\bigr)\Bigr)\cdot\rho_{\alpha\beta}\bigl(\mathbf{x}'+[\lambda^{-1}]_\alpha\bigr) - \lambda \tau(\mathbf{x})\rho_{\alpha\beta}(\mathbf{x}'),$$
where we have used the formula \eqref{formula}.
So finally if we set $\mathbf{x} \to \mathbf{x} + [\lambda^{-1}]_\alpha, \, \mathbf{x}' \to \mathbf{x}$, then we can obtain \eqref{startwo}.
\end{proof}
Now let us prove {\bf Theorem \ref{theoremtau}}. Firstly, it is obviously that \eqref{tau11}-\eqref{tau3} follows from
\eqref{qandrbilinearequation1}-\eqref{rbilinearequation1} and \eqref{wavematrix1}
\eqref{Omegaandtau1} \eqref{Omegaandtau2}, which is the first part of Theorem \ref{theoremtau}.
We now prove the second part of Theorem \ref{theoremtau}. For $\tau_{\alpha\beta}(\mathbf{x})$ with $\tau_{\alpha\alpha}(\mathbf{x})=\tau(\mathbf{x})$, $\rho_{i,\alpha\beta} (\mathbf{x}),\, \sigma_{i,\alpha\beta} (\mathbf{x})$
satisfy \eqref{tau11}-\eqref{tau3}, let us firstly set
\begin{align*}
&\widetilde{\Psi}_{\alpha\beta}(\mathbf{x}, z)=\varepsilon_{\alpha\beta}\frac{\tau_{\alpha\beta}\left(\mathbf{x}-[z^{-1}]_\beta\right)}{\tau(\mathbf{x})}z^{\delta_{\alpha\beta}-1}e^{\xi(\mathbf{x}^{(\beta)}, z)}, \\
&\widetilde{\Psi}_{\beta\alpha}^{*}(\mathbf{x},z)=\varepsilon_{\beta\alpha}\frac{\tau_{\alpha\beta}\left(\mathbf{x}+[z^{-1}]_\alpha\right)}{\tau(\mathbf{x})}z^{\delta_{\beta\alpha}-1} e^{-\xi(\mathbf{x}^{(\alpha)}, z)},\\
&\mathcal{Q}_{i,\alpha\beta}(\mathbf{x})= \frac{\rho_{i,\alpha\beta} (\mathbf{x})}{\tau(\mathbf{x})}, \quad \mathcal{R}_{i,\alpha\beta}(\mathbf{x})=\frac{\sigma_{i,\alpha\beta}(\mathbf{x})}{\tau(\mathbf{x})}, \quad 1\leq i\leq m.
\end{align*}
Then by \eqref{tau11}-\eqref{tau3}, we have
\begin{align}
&\operatorname{Res}_z z^k \sum_{\gamma=1}^{N}\widetilde{\Psi}_{\alpha\gamma}(\mathbf{x},z)\widetilde{\Psi}^*_{\beta\gamma}(\mathbf{x}',z)=\sum_{\gamma=1}^{N}\sum_{i=1}^m Q_{i,\alpha\gamma}(\mathbf{x})R_{i,\beta\gamma}(\mathbf{x}), \nonumber\\
&\operatorname{Res}_z z^{-1} \sum_{\gamma=1}^{N} \widetilde{\Psi}_{\alpha\gamma}(\mathbf{x},z)\widetilde{\Psi}^*_{\gamma\gamma}(\mathbf{x}',z)Q_{i,\alpha\beta}\big(\mathbf{x}'+[z^{-1}]_\gamma\big)=Q_{i,\alpha\beta}(\mathbf{x}), \label{D2} \\
&\operatorname{Res}_z z^{-1} \sum_{\gamma=1}^{N} \widetilde{\Psi}_{\gamma\gamma}(\mathbf{x},z)\widetilde{\Psi}^*_{\beta\gamma}(\mathbf{x}',z)R_{i,\alpha\beta}\big(\mathbf{x}-[z^{-1}]_\gamma\big)=R_{i,\alpha\beta}(\mathbf{x}'). \label{D3}
\end{align}
Furthermore from Lemma \ref{D}, we conclude that

\begin{align*}
&\sum_{\gamma=1}^{N} \widetilde{\Psi}^*_{\gamma\alpha}(\mathbf{x},z)Q_{\gamma\beta}(\mathbf{x})=-z^{-1}D\big(\widetilde{\Psi}^*_{\alpha\alpha}(\mathbf{x},z)\mathcal{Q}( \mathbf{x} + [z^{-1}]_\alpha)_{\alpha\beta}\big), \\
&\sum_{\gamma=1}^{N} R_{\gamma\alpha}(\mathbf{x})\widetilde{\Psi}_{\gamma\beta}(\mathbf{x},z)=z^{-1}D\big(\widetilde{\Psi}_{\beta\beta}(\mathbf{x},z)\mathcal{R}( \mathbf{x} - [z^{-1}]_\beta )_{\beta\alpha}\big),
\end{align*}
which imply
\begin{align*}
\Omega(\mathcal{Q}(\mathbf{x}),\widetilde{\Psi}^*(\mathbf{x},z))_{\alpha\beta} = -z^{-1}\widetilde{\Psi}^*_{\alpha\alpha}(\mathbf{x},z)\mathcal{Q}\left( \mathbf{x} + [z^{-1}]_\alpha \right)_{\alpha\beta}, \quad\Omega(\widetilde{\Psi}(\mathbf{x},z),\mathcal{R}(\mathbf{x}))_{\alpha\beta} =z^{-1}\widetilde{\Psi}_{\beta\beta}(\mathbf{x},z)\mathcal{R}\left( \mathbf{x} - [z^{-1}]_\beta \right)_{\beta\alpha}.
\end{align*}
Thus \eqref{D2} \eqref{D3} can be rewritten as \eqref{qbilinearequation1}
\eqref{rbilinearequation1}. So far, all conditions required for the second
part of Theorem \ref{theorembilinear} have been derived, namely
\eqref{qandrbilinearequation1}-\eqref{rbilinearequation1}. Consequently the second part of Theorem \ref{theoremtau} can be proved based on the Theorem
\ref{theorembilinear}. By now, we have finished the proof of {\bf Theorem \ref{theoremtau}}.

\subsection{Hirota bilinear equations} For the results in Theorem \ref{theoremtau}, it will be more convenient to use the  Hirota bilinear operator $D$ \cite{hirota2004} to rewrite the bilinear equations \eqref{tau11}-\eqref{tau3}. Here the Hirota bilinear operator is defined by
$$f(D_{\bf x}) g({\bf x})\cdot h({\bf x})=f({\partial_{\bf x}})(g({\bf x}+{\bf y})h({\bf x}+{\bf y}))|_{{\bf y}=0},$$
where $f$ is the function of $D_{\bf x}=(D_{{\bf x}^{(1)}},\ldots,D_{{\bf x}^{(N)}})$ with $D_{{\bf x}^{(\gamma)}}=(D_{x_1^{(\gamma)}},D_{x_2^{(\gamma)}},\ldots)$.
 By direct computations, one can rewrite \eqref{tau11}-\eqref{tau3} into
\begin{itemize}
	\item Hirota bilinear form of \eqref{tau11}
	\begin{equation}\label{Hirotatau1}
		\begin{aligned}
			&\sum_{\gamma=1}^N \varepsilon_{\alpha\gamma}\varepsilon_{\beta\gamma}
			\exp\left(
			\sum_{l=1}^{+\infty} \sum_{\eta=1}^N y_l^{(\eta)} D_{x_l^{(\eta)}}
			\right)
			\sum_{q = 0}^{+\infty} p_q(-2y^{(\gamma)})
			\cdot
			p_{k+\delta_{\alpha\gamma}+\delta_{\beta\gamma}+q-1}\bigl(\widetilde{D}_{\mathbf{x}^{(\gamma)}}\bigr)
			\tau_{\gamma\beta}(\mathbf{x})\cdot \tau_{\alpha\gamma}(\mathbf{x})\\
			&=\sum_{\gamma=1}^N \sum_{i=1}^m
			\exp\left(
			\sum_{l=1}^{+\infty} \sum_{\eta=1}^N y_l^{(\eta)} D_{x_l^{(\eta)}}
			\right)
			\sigma_{i,\beta\gamma}(\mathbf{x})\cdot \rho_{i,\alpha\gamma}(\mathbf{x}),
		\end{aligned}
	\end{equation}
	\item Hirota bilinear form of \eqref{tau2}	
	\begin{equation}\label{Hirotatau2}
		\begin{aligned}
			 &\sum_{\gamma=1}^N \varepsilon_{\alpha\gamma}\exp\left(
			\sum_{l=1}^{+\infty} \sum_{\eta=1}^N y_l^{(\eta)} D_{x_l^{(\eta)}}
			\right)
			\sum_{q = 0}^{+\infty} p_q(-2y^{(\gamma)})
			\cdot
			p_{\delta_{\alpha\gamma}+q-1}\bigl(\widetilde{D}_{\mathbf{x}^{(\gamma)}}\bigr)
			\rho_{i, \gamma\beta}(\mathbf{x})\cdot \tau_{\alpha\gamma}(\mathbf{x})\\
			&=\exp\left(
			\sum_{l=1}^{+\infty} \sum_{\eta=1}^N y_l^{(\eta)} D_{x_l^{(\eta)}}
			\right)\tau(\mathbf{x})
			\cdot \rho_{i,\alpha\beta}(\mathbf{x}),
		\end{aligned}
	\end{equation}
	\item Hirota bilinear form of \eqref{tau3}
	\begin{equation}\label{Hirotatau3}
		\begin{aligned}
			&\sum_{\gamma=1}^N \varepsilon_{\beta\gamma}\exp\left(
			\sum_{l=1}^{+\infty} \sum_{\eta=1}^N y_l^{(\eta)} D_{x_l^{(\eta)}}
			\right)
			\sum_{q = 0}^{+\infty} p_q(-2y^{(\gamma)})
			\cdot
			p_{\delta_{\beta\gamma}+q-1}\bigl(\widetilde{D}_{\mathbf{x}^{(\gamma)}}\bigr)
			\tau_{\gamma\beta}(\mathbf{x})\cdot \sigma_{i,\gamma\alpha}(\mathbf{x})\\
			&=\exp\left(
			\sum_{l=1}^{+\infty} \sum_{\eta=1}^N y_l^{(\eta)} D_{x_l^{(\eta)}}
			\right) \sigma_{i,\beta\alpha}(\mathbf{x})\cdot\tau(\mathbf{x}).
		\end{aligned}
	\end{equation}
\end{itemize}
Here $p_{n}(\mathbf{x}^{(\gamma)})$ is the Schur polynomial defined by $e^{\xi(\mathbf{x}^{(\gamma)},\lambda)}=\sum_{n=0}^{\infty}p_{n}(\mathbf{x}^{(\gamma)})\lambda^{n}$ and $p_{n}(\mathbf{x}^{(\gamma)})=0$ if $n<0$.

If we consider the $(1,1)$-constrained $2$-component KP hierarchy, that is we have set $k=m=1$, $N=2$ in \eqref{Hirotatau1}-\eqref{Hirotatau3}, then we can obtain the following examples of the Hirota bilinear equations. \\
$\bullet$ {\bf Example } of \eqref{Hirotatau1}:
\begin{align*}
	&\frac{1}{2}\left(D_{x_1^{(1)}}^2+ D_{x_2^{(1)}}\right)\tau\cdot\tau
+\tau_{21}\cdot\tau_{12}
=\sigma_{11}\cdot\rho_{11}+\sigma_{12}\cdot\rho_{12},\\
	&-D_{x_1^{(1)}}\tau_{12}\cdot\tau
+D_{x_1^{(2)}}\tau\cdot\tau_{12}
=\sigma_{21}\cdot\rho_{11}+\sigma_{22}\cdot\rho_{12},\\
	&-D_{x_1^{(1)}}\tau\cdot\tau_{21}
+D_{x_1^{(2)}}\tau_{21}\cdot\tau
=\sigma_{11}\cdot\rho_{21}+\sigma_{12}\cdot\rho_{22},\\
	&\frac{1}{2}\left(D_{x_1^{(2)}}^2+D_{x_2^{(2)}}\right)\tau\cdot\tau
+\tau_{12}\cdot\tau_{21}
=\sigma_{21}\cdot\rho_{21}+\sigma_{22}\cdot\rho_{22},
\end{align*}                                   \begin{align*}
	&\frac{1}{6}\left(D_{x_1^{(1)}}^3-3D_{x_1^{(1)}}D_{x_2^{(1)}}-4D_{x_3^{(1)}}\right)\tau\cdot\tau
+D_{x_1^{(1)}}\tau_{21}\cdot\tau_{12}
=D_{x_1^{(1)}}\bigl(\sigma_{11}\cdot\rho_{11}+\sigma_{12}\cdot\rho_{12}\bigr),\\
	&D_{x_1^{(1)}}D_{x_1^{(2)}}\tau\cdot\tau_{12}+D_{x_2^{(1)}}\tau_{12}\cdot\tau
=D_{x_1^{(1)}}\bigl(\sigma_{21}\cdot\rho_{11}+\sigma_{22}\cdot\rho_{12}\bigr),\\
	&D_{x_1^{(1)}}D_{x_1^{(2)}}\tau_{21}\cdot\tau+D_{x_2^{(1)}}\tau\cdot\tau_{21}
=D_{x_1^{(1)}}\bigl(\sigma_{11}\cdot\rho_{21}+\sigma_{12}\cdot\rho_{22}\bigr),\\
	&\frac{1}{2}\left(D_{x_1^{(2)}}^2+D_{x_2^{(2)}}\right)\tau\cdot\tau
-D_{x_1^{(1)}}\tau_{12}\cdot\tau_{21}
=D_{x_1^{(1)}}\bigl(\sigma_{21}\cdot\rho_{21}+\sigma_{22}\cdot\rho_{22}\bigr),
\end{align*}

$\bullet$ {\bf Example } of \eqref{Hirotatau2}:

\begin{align*}
	&D_{x_{1}^{(1)}}\tau\cdot\rho_{21}-\tau_{21}\cdot\rho_{11}=0,\\
	&D_{x_{1}^{(1)}}\tau\cdot\rho_{22}-\tau_{21}\cdot\rho_{12}=0,
\end{align*}
\begin{align*}
	&D_{x_{1}^{(1)}}\rho_{21}\cdot\tau_{12} + D_{x_{1}^{(2)}}D_{x_{1}^{(1)}}\rho_{11}\cdot\tau=0,\\
	&D_{x_{1}^{(1)}}\rho_{22}\cdot\tau_{12} + D_{x_{1}^{(2)}}D_{x_{1}^{(1)}}\rho_{12}\cdot\tau=0,\\
	&D_{x_{1}^{(2)}}\rho_{11}\cdot\tau_{21} - D_{x_{1}^{(2)}}D_{x_{1}^{(1)}}\rho_{21}\cdot\tau=0,\\
	&D_{x_{1}^{(2)}}\rho_{12}\cdot\tau_{21} - D_{x_{1}^{(2)}}D_{x_{1}^{(1)}}\rho_{22}\cdot\tau=0.
\end{align*}

$\bullet$ {\bf Example } of \eqref{Hirotatau3}:
\begin{align*}
	&D_{x_{1}^{(1)}}\tau\cdot\sigma_{21} + \tau_{12}\cdot\sigma_{11}=0,\\
	&D_{x_{1}^{(1)}}\tau\cdot\sigma_{22} + \tau_{12}\cdot\sigma_{12}=0.
\end{align*}
\begin{align*}
	&D_{x_{1}^{(1)}}\tau_{21}\cdot \sigma_{21}+ D_{x_{1}^{(2)}}D_{x_{1}^{(1)}}\tau\cdot\sigma_{11}=0,\\
	&D_{x_{1}^{(1)}}\tau_{21}\cdot\sigma_{22} + D_{x_{1}^{(2)}}D_{x_{1}^{(1)}}\tau\cdot\sigma_{12}=0,\\
	&D_{x_{1}^{(2)}}\tau_{12}\cdot\sigma_{11} - D_{x_{1}^{(2)}}D_{x_{1}^{(1)}}\tau\cdot\sigma_{21}=0,\\
	&D_{x_{1}^{(2)}}\tau_{12}\cdot\sigma_{12} - D_{x_{1}^{(2)}}D_{x_{1}^{(1)}}\tau\cdot\sigma_{22}=0.
\end{align*}

\section{Solutions of the constrained matrix KP hierarchy}
In this section, we will construct the solutions of the constrained matrix KP hierarchy by using the multi-component boson-fermion correspondence. Firstly, we will list some important properties of the operators $S_{n}$ \cite{kac2023,kac1993}.
\begin{align*}
S_n = \sum_{\alpha=1}^N S_n^{(\alpha)},\quad S_n^{(\alpha)} = \operatorname{Res}_z z^{n-1} \psi^{(\alpha)}(z) \otimes \psi^{*(\alpha)}(z)
= \sum_j \psi_j^{(\alpha)} \otimes \psi_{j+n}^{*(\alpha)}, \quad n \in \mathds{Z},
\end{align*}
where $\psi^{(\alpha)}(z)=\sum_{j\in\mathbb{Z}}\psi^{(\alpha)}_{j}z^{j}$ and $\psi^{*(\alpha)}(z)=\sum_{j\in\mathbb{Z}}\psi^{*(\alpha)}_{j}z^{-j}$.
Then the tau functions of the constrained $N$-component KP hierarchy expressed by the vacuum expectation value of multi-component fermions. Finally the solutions of the constrained matrix KP hierarchy are obtained, which is just Theorem \ref{theorem:solution}.

\subsection{Properties of operators $S_n$}
 The operators $S_{n}$ are very important in seeking the solutions of bilinear equations. Here the corresponding properties of $S_{n}$ are given in the below lemmas, which can be obtained by direct computations.

\begin{lemma}\label{lemmafemion}
Given $\beta^{(j)}\in V^{(j)}=\bigoplus_{i\in \mathds{Z}}\mathbb{C}\psi_{i}^{(j)}$ and $\beta^{*(j)}\in V^{*(j)}=\bigoplus_{i\in \mathbb{Z}}\mathbb{C}\psi_{i}^{*(j)}$, $S_n^{(\alpha)}$ satisfies the following identities:
\begin{align*}
&S_{n}^{(\alpha)}(1 \otimes \beta^{(j)})= \delta_{\alpha j}\beta_{[n]}^{(j)}\otimes 1 - (1\otimes \beta^{(j)}) S_{n}^{(\alpha)},\\
&S_{n}^{(\alpha)}(\beta^{(j)} \otimes 1)= -(\beta^{(j)} \otimes 1)S_{n}^{(\alpha)},\\
&S_{n}^{(\alpha)} \bigl( \beta^{(i)} \otimes \beta^{(j)} \bigr) = \delta_{\alpha j}\beta_{[n]}^{(j)} \beta^{(i)} \otimes 1 + \bigl( \beta^{(i)} \otimes \beta^{(j)} \bigr) S_{n}^{(\alpha)},
\end{align*}
and
\begin{align*}
&S_n^{(\alpha)}(\beta^{*(j)} \otimes 1)= \delta_{\alpha j}1 \otimes \beta^{*(j)}_{[-n]} - (\beta^{*(j)}\otimes 1) S_n^{(\alpha)},\\
&S_n^{(\alpha)}(1 \otimes \beta^{*(j)})= -(1 \otimes \beta^{*(j)})S_n^{(\alpha)},\\
&S_n^{(\alpha)}(\beta^{*(i)} \otimes \beta^{*(j)})= \delta_{\alpha j}1 \otimes \beta^{*(i)}_{[-n]} \beta^{*(j)}  + (\beta^{*(i)} \otimes \beta^{*(j)}) S_n^{(\alpha)},
\end{align*}
where $\beta^j_{[n]}$ and $\beta^{*j}_{[-n]}$ are given in \eqref{betadefination}.
\end{lemma}
\begin{lemma}\label{totalSn}
Given $\beta\in V=\bigoplus_{\alpha=1}^{N}V^{(\alpha)}, \, \beta^{*}\in V^{*}=\bigoplus_{\alpha=1}^{N}V^{*(\alpha)}$, we have
\begin{align*}
&S_n(1 \otimes \beta)= \beta_{[n]}\otimes 1 - (1\otimes \beta) S_n,\\
&S_n(\beta \otimes 1)= -(\beta \otimes 1)S_n,\\
&S_n(\beta \otimes \beta)= \beta_{[n]} \beta \otimes 1 + (\beta \otimes \beta) S_n,
\end{align*}
and
\begin{align*}
&S_n(\beta^{*} \otimes 1)= 1 \otimes \beta^{*}_{[-n]} - (\beta^{*}\otimes 1) S_n,\\
&S_n(1 \otimes \beta^{*})= -(1 \otimes \beta^{*})S_n,\\
&S_n(\beta^{*} \otimes \beta^{*})= 1 \otimes \beta^{*}_{[-n]} \beta^{*}  + (\beta^{*} \otimes \beta^{*}) S_n.
\end{align*}
\end{lemma}
\begin{proposition}\label{propositionS}
Given $\beta_{i}\in V$ and $\beta_{i}^{*}\in V^{*}$, we have the following identities:
\begin{equation}
\begin{aligned}\label{fermionS}
S_n\bigl( \beta^{*[\vec{\mathcal{N}}]} \beta^{[\vec{\mathcal{M}}]} \otimes \beta^{*[\vec{\mathcal{N}}]} \beta^{[\vec{\mathcal{M}}]} \bigr)
&=\sum_{i=1}^{\mathcal{M}} \beta^{*[\vec{\mathcal{N}}]} \beta^{[\vec{\mathcal{M}}\setminus\vec{i}]} \beta_{i[n]} \beta^{[\vec{i}]} \otimes \beta^{*[\vec{\mathcal{N}}]} \beta^{[\vec{\mathcal{M}}\setminus\{i\}]} \\
&+\sum_{j=1}^{\mathcal{N}} \beta^{*[\vec{\mathcal{N}}\setminus\{j\}]}\beta^{[\vec{\mathcal{M}}]} \otimes \beta^{*[\vec{\mathcal{N}}\setminus\vec{j}]} \beta_{j[-n]}^{*} \beta^{*[\vec{j}]} \beta^{[\vec{\mathcal{M}}]} \\
&+\bigl( \beta^{*[\vec{\mathcal{N}}]} \beta^{[\vec{\mathcal{M}}]} \otimes \beta^{*[\vec{\mathcal{N}}]} \beta^{[\vec{\mathcal{M}}]} \bigr) S_n,
\end{aligned}
\end{equation}
where the symbols $\beta^{[\vec{\mathcal{M}}]}$ and $\beta^{*[\vec{\mathcal{N}}]}$ can be found in \eqref{signbb*}.
\end{proposition}

\begin{proof}
First, \eqref{fermionS} is obviously correct for $\mathcal{M}+\mathcal{N}=1$ by Lemma
\ref{totalSn}. Next let us assume \eqref{fermionS} is correct for $\mathcal{M}+\mathcal{N}$,
then we will prove that the case of $\mathcal{M}+\mathcal{N}+1$ also satisfies \eqref{fermionS}.
From case $\mathcal{M}+\mathcal{N}$ to case $\mathcal{M}+\mathcal{N}+1$, there are two ways. One is given by
$\beta^{\ast [\overrightarrow{\mathcal{N}}]}\beta^{[\overrightarrow{\mathcal{M}+1}\setminus\{1\}]}\beta_{1}$ and
another is $\beta^{*}_{\mathcal{N}+1}\beta^{*[\overrightarrow{\mathcal{N}}]} \beta^{[\overrightarrow{\mathcal{M}}]}$. Here
we only consider the first way, since the second is almost the same. In fact, by the assumption that \eqref{fermionS} holds for $\mathcal{M}+\mathcal{N}$, we have

\begin{align*}
&S_n\bigl( \beta^{*[\overrightarrow{\mathcal{N}}]} \beta^{[\overrightarrow{\mathcal{M}+1}]} \otimes \beta^{*[\overrightarrow{\mathcal{N}}]} \beta^{[\overrightarrow{\mathcal{M}+1}]} \bigr)
= S_n\Bigl( \beta^{\ast [\overrightarrow{\mathcal{N}}]}\beta^{[\overrightarrow{\mathcal{M}+1}\setminus\{1\}]} \otimes \beta^{\ast [\overrightarrow{\mathcal{N}}]}\beta^{[\overrightarrow{\mathcal{M}+1}\setminus\{1\}]} \Bigr) (\beta_1 \otimes \beta_1) \\
&=\sum_{i=2}^{\mathcal{M}+1} \beta^{*[\overrightarrow{\mathcal{N}}]} \beta^{[\overrightarrow{\mathcal{M}+1}\setminus\overrightarrow{i}]} \beta_{i[n]} \beta^{[\overrightarrow{i}]} \otimes \beta^{*[\overrightarrow{\mathcal{N}}]} \beta^{[\overrightarrow{\mathcal{M}+1}\setminus\{i\}]}+\sum_{j=1}^{\mathcal{N}} \beta^{*[\overrightarrow{\mathcal{N}}\setminus\{j\}]}\beta^{[\overrightarrow{\mathcal{M}+1}]} \otimes \beta^{*[\overrightarrow{\mathcal{N}}\setminus\overrightarrow{j}]} \beta_{j[-n]}^{*} \beta^{*[\overrightarrow{j}]} \beta^{[\overrightarrow{\mathcal{M}+1}]} \\
&+\bigl( \beta^{*[\overrightarrow{\mathcal{N}}]} \beta^{[\overrightarrow{\mathcal{M}+1}\setminus\{1\}]} \otimes \beta^{*[\overrightarrow{\mathcal{N}}]} \beta^{[\overrightarrow{\mathcal{M}+1}\setminus\{1\}]} \bigr) S_n (\beta_1 \otimes \beta_1)\\
&=\sum_{i=1}^{\mathcal{M}+1} \beta^{*[\overrightarrow{\mathcal{N}}]} \beta^{[\overrightarrow{\mathcal{M}+1}\setminus\overrightarrow{i}]} \beta_{i[n]} \beta^{[\overrightarrow{i}]} \otimes \beta^{*[\overrightarrow{\mathcal{N}}]} \beta^{[\overrightarrow{\mathcal{M}+1}\setminus\{i\}]} +\sum_{j=1}^{\mathcal{N}} \beta^{*[\overrightarrow{\mathcal{N}}\setminus\{j\}]}\beta^{[\overrightarrow{\mathcal{M}+1}]} \otimes \beta^{*[\overrightarrow{\mathcal{N}}\setminus\overrightarrow{j}]} \beta_{j[-n]}^{*} \beta^{*[\overrightarrow{j}]} \beta^{[\overrightarrow{\mathcal{M}+1}]} \\
&+\bigl( \beta^{*[\overrightarrow{\mathcal{N}}]} \beta^{[\overrightarrow{\mathcal{M}+1}]} \otimes \beta^{*[\overrightarrow{\mathcal{N}}]} \beta^{[\overrightarrow{\mathcal{M}+1}]} \bigr) S_n.
\end{align*}
Here in the first equality we use $(a\otimes b)(c\otimes d)=(ac\otimes bd)$, and in the
second equality we use \eqref{fermionS} for $\mathcal{M}+\mathcal{N}$. Thus
the conclusion \eqref{fermionS} still holds for $\mathcal{M}+\mathcal{N}+1$.
\end{proof}
\begin{corollary}\label{corollary:Sk}
Given $\tau^{[\overrightarrow{\mathcal{M}},\overrightarrow{\mathcal{N}}]}$ in \eqref{ABCD}, we have
\begin{align}\label{Sk}
S_k\big(\tau^{[\overrightarrow{\mathcal{M}},\overrightarrow{\mathcal{N}}]}|0\rangle\otimes \tau^{[\overrightarrow{\mathcal{M}},\overrightarrow{\mathcal{N}}]}|0\rangle\big)
=\sum_{i=1}^P \sum_{\gamma=1}^N A_{i\gamma}^{[\overrightarrow{\mathcal{M}},\overrightarrow{\mathcal{N}}]}|0\rangle  \otimes B_{i\gamma}^{[\overrightarrow{\mathcal{M}},\overrightarrow{\mathcal{N}}]}|0\rangle+\sum_{i=1}^G \sum_{\gamma=1}^N C_{i\gamma}^{[\overrightarrow{\mathcal{M}},\overrightarrow{\mathcal{N}}]}|0\rangle\otimes
D_{i\gamma}^{[\overrightarrow{\mathcal{M}},\overrightarrow{\mathcal{N}}]}|0\rangle,
\end{align}
where $A_{i\gamma}^{[\overrightarrow{\mathcal{M}},\overrightarrow{\mathcal{N}}]}, B_{i\gamma}^{[\overrightarrow{\mathcal{M}},\overrightarrow{\mathcal{N}}]}, C_{i\gamma}^{[\overrightarrow{\mathcal{M}},\overrightarrow{\mathcal{N}}]}, D_{i\gamma}^{[\overrightarrow{\mathcal{M}},\overrightarrow{\mathcal{N}}]}$ are given in \eqref{ABCD}.
\end{corollary}
\begin{proof}
From Proposition \ref{propositionS} and $S_{k}(|0\rangle \otimes|0\rangle)=0$ derived by \eqref{vacuumdefinition}, we can obtain
\begin{equation}
\begin{aligned}\label{skmn}
S_k\big(\beta^{*[\vec{\mathcal{N}}]}\beta^{[\vec{\mathcal{M}}]}|0\rangle \otimes \beta^{*[\vec{\mathcal{N}}]}\beta^{[\vec{\mathcal{M}}]}|0\rangle\big)
&=\sum_{j=1}^{\mathcal{M}} \beta^{*[\vec{\mathcal{N}}]} \beta^{[\vec{\mathcal{M}}\setminus\vec{j}]} \beta_{j[k]} \beta^{[\vec{j}]} |0\rangle\otimes \beta^{*[\vec{\mathcal{N}}]} \beta^{[\vec{\mathcal{M}}\setminus\{j\}]}|0\rangle \\
&+\sum_{l=1}^{\mathcal{N}} \beta^{*[\vec{\mathcal{N}}\setminus\{l\}]}\beta^{[\vec{\mathcal{M}}]} |0\rangle\otimes \beta^{*[\vec{\mathcal{N}}\setminus\vec{l}]} \beta_{l[-k]}^{*} \beta^{*[\vec{l}]} \beta^{[\vec{\mathcal{M}}]}|0\rangle.
\end{aligned}
\end{equation}
When $1\leq j\leq \mathcal{M}-KN$ and $j\neq p(i,\gamma)$, $\beta^{[\vec{\mathcal{M}}\setminus\vec{j}]}$ will contain $\beta_{j[k]}$, the corresponding term equals zero by $\beta_{j[k]}^{2}=0$. For $1+\mathcal{M}-KN\leq j\leq \mathcal{M}$, the corresponding term is also zero due to $\Phi_{aj,[k]}\Phi_{aj}=0$. Therefore, the first term in the right hand side of \eqref{skmn} reduces to $\sum_{i=1}^P \sum_{\gamma=1}^N A_{i\gamma}^{[\overrightarrow{\mathcal{M}},\overrightarrow{\mathcal{N}}]}|0\rangle  \otimes B_{i\gamma}^{[\overrightarrow{\mathcal{M}},\overrightarrow{\mathcal{N}}]}|0\rangle.$
We can similarly deal with the second term in the right hand side of \eqref{skmn}.
\end{proof}
\subsection{Proof of Theorem \ref{theorem:solution}}
If we apply $ \langle \mathbf{p}+\mathbf{e}_\alpha | e^{H(\mathbf{x})} \otimes
\langle \mathbf{p}-\mathbf{e}_\beta | e^{H(\mathbf{x}')}$ to both sides in
\eqref{Sk}, then by
$$S_k={\rm Res}_zz^{k-1}\sum_{\alpha=1}^N\psi^{(\alpha)}(z) \otimes \psi^{*(\alpha)}(z)$$
we can know
\begin{equation}
\begin{aligned}\label{pbilinear}
\operatorname{Res}_z z^{-1} \sum_{\gamma=1}^{N} \langle \mathbf{p}+\mathbf{e}_{\alpha} | e^{H(\mathbf{x})} \psi^{(\gamma)}(z) \tau^{[\overrightarrow{\mathcal{M}},\overrightarrow{\mathcal{N}}]} |0\rangle \cdot\langle \mathbf{p}-\mathbf{e}_\beta | e^{H(\mathbf{x}')} \psi^{*(\gamma)}(z) \tau^{[\overrightarrow{\mathcal{M}},\overrightarrow{\mathcal{N}}]} |0 \rangle\\
=\sum_{i=1}^P \sum_{\gamma=1}^N \langle \mathbf{p}+\mathbf{e}_{\alpha} | e^{H(\mathbf{x})}A_{i\gamma}^{[\overrightarrow{\mathcal{M}},\overrightarrow{\mathcal{N}}]}|0\rangle \cdot \langle \mathbf{p}-\mathbf{e}_\beta | e^{H(\mathbf{x}')} B_{i\gamma}^{[\overrightarrow{\mathcal{M}},\overrightarrow{\mathcal{N}}]}|0\rangle\\
+\sum_{i=1}^G \sum_{\gamma=1}^N \langle \mathbf{p}+\mathbf{e}_{\alpha} | e^{H(\mathbf{x})} C_{i\gamma}^{[\overrightarrow{\mathcal{M}},\overrightarrow{\mathcal{N}}]}|0\rangle \cdot \langle \mathbf{p}-\mathbf{e}_\beta | e^{H(\mathbf{x}')}
D_{i\gamma}^{[\overrightarrow{\mathcal{M}},\overrightarrow{\mathcal{N}}]}|0\rangle.
\end{aligned}
\end{equation}
Further by the following relations \cite{kac2023,zabrodin2019}:
\begin{align*}
&\langle \mathbf{p} + \mathbf{e}_\alpha | e^{H(\mathbf{x})}\psi^{(\gamma)}(z) =\varepsilon_{\alpha\gamma} \varepsilon_\gamma(\mathbf{p})z^{p_\gamma + \delta_{\alpha\gamma} - 1} e^{\xi(x^{(\gamma)},z)}\langle \mathbf{p} + \mathbf{e}_\alpha - \mathbf{e}_\gamma | e^{H(\mathbf{x} - [z^{-1}]_\gamma)},\\
&\langle \mathbf{p} - \mathbf{e}_\beta | e^{H(\mathbf{x})}\psi^{*(\gamma)}(z) =\varepsilon_{\beta \gamma} \varepsilon_\gamma(\mathbf{p})z^{-p_\gamma + \delta_{\beta \gamma}} e^{-\xi(x^{(\gamma)},z)}\langle \mathbf{p} - \mathbf{e}_\beta +\mathbf{e}_\gamma | e^{H(\mathbf{x} + [z^{-1}]_\gamma)},
\end{align*}
where $\varepsilon_\gamma(\mathbf{p})=(-1)^{\sum_{l=\gamma+1}^{N}p_{l}}$,
we can find $\tau_{\alpha\beta}(\mathbf{x})$, $\rho_{i,\alpha
\gamma}(\mathbf{x})$ and $\sigma_{i,\beta \gamma}(\mathbf{x})$ in \eqref{solution1}-\eqref{solution3} satisfy the bilinear equation \eqref{tau11}.

When $1 \le i \le P$, we have $S_0(\tau^{[\overrightarrow{\mathcal{M}},\overrightarrow{\mathcal{N}}]} |0\rangle\otimes A_{i\gamma}^{[\overrightarrow{\mathcal{M}},\overrightarrow{\mathcal{N}}]}|0\rangle) = A_{i\gamma}^{[\overrightarrow{\mathcal{M}},\overrightarrow{\mathcal{N}}]}|0\rangle
\otimes \tau^{[\overrightarrow{\mathcal{M}},\overrightarrow{\mathcal{N}}]}|0\rangle$, that is
\[
\operatorname{Res}_z z^{-1} \sum_{\gamma=1}^N \psi^{(\gamma)}(z) \tau^{[\overrightarrow{\mathcal{M}},\overrightarrow{\mathcal{N}}]} |0\rangle \otimes \psi^{*(\gamma)}(z) A_{i\gamma}^{[\overrightarrow{\mathcal{M}},\overrightarrow{\mathcal{N}}]} |0\rangle
= A_{i\gamma}^{[\overrightarrow{\mathcal{M}},\overrightarrow{\mathcal{N}}]} |0\rangle \otimes \tau^{[\overrightarrow{\mathcal{M}},\overrightarrow{\mathcal{N}}]} |0\rangle.
\]
 Applying $ \langle \mathbf{p}+\mathbf{e}_\alpha | e^{H(\mathbf{x})} \otimes \langle \mathbf{p} | e^{H(\mathbf{x}')}$ to the above formula
yields
\begin{align*}
\operatorname{Res}_z z^{-1} \sum_{\gamma=1}^{N} \langle \mathbf{p}+\mathbf{e}_{\alpha} | e^{H(\mathbf{x})} \psi^{(\gamma)}(z) \tau^{[\overrightarrow{\mathcal{M}},\overrightarrow{\mathcal{N}}]} |0\rangle \cdot\langle \mathbf{p} | e^{H(\mathbf{x}')} \psi^{*(\gamma)}(z)  A_{i\gamma}^{[\overrightarrow{\mathcal{M}},\overrightarrow{\mathcal{N}}]} |0 \rangle\\
= \langle \mathbf{p}+\mathbf{e}_{\alpha} | e^{H(\mathbf{x})}\tau^{[\overrightarrow{\mathcal{M}},\overrightarrow{\mathcal{N}}]} |0\rangle \cdot\langle \mathbf{p} | e^{H(\mathbf{x}')} A_{i\gamma}^{[\overrightarrow{\mathcal{M}},\overrightarrow{\mathcal{N}}]} |0 \rangle,
\end{align*}
then by \eqref{pbilinear} and the relations \cite{kac2023,zabrodin2019}
\begin{align*}
&\langle \mathbf{p} | e^{H(\mathbf{x})}\psi^{(\gamma)}(z) =e^{\xi(x^{(\gamma)},z)}\varepsilon_{\gamma}(\mathbf{p})z^{p_{\gamma}-1}\langle \mathbf{p} - \mathbf{e}_\gamma | e^{H(\mathbf{x}- [z^{-1}]_\gamma)}\\
&\langle \mathbf{p} | e^{H(\mathbf{x})}\psi^{*(\gamma)}(z) =e^{-\xi(x^{(\gamma)},z)}\varepsilon_{\gamma}(\mathbf{p})z^{-p_{\gamma}}\langle \mathbf{p} + \mathbf{e}_\gamma | e^{H(\mathbf{x} + [z^{-1}]_\gamma)}
\end{align*}
we can prove that $\tau_{\alpha\beta}$ and $\rho_{i,\alpha\beta}$ in \eqref{solution1}-\eqref{solution3} will satisfy \eqref{tau2}. Other cases can be similarly proved.

\section{Examples of constrained matrix KP hierarchy}
In this section, we will give some examples of the constrained matrix KP hierarchy and give the corresponding solutions.
\subsection{$(1,1)$-$2\times2$ matrix KP hierarchy}
Set $k=m=1, N=2$ in \eqref{con-matrix-KP-Lax-operator}, then
$$L(\mathbf{t},\partial)=\partial+Q(\mathbf{t})\partial^{-1}R(\mathbf{t})^{\intercal},$$
where
$
Q=
\begin{pmatrix}
Q_{11} & Q_{12} \\
Q_{21} & Q_{22}
\end{pmatrix}
$
and
$
R =
\begin{pmatrix}
R_{11} & R_{12} \\
R_{21} & R_{22}
\end{pmatrix}.
$
Comparing with $L(\mathbf{t},\partial) = \partial + \sum_{i=1}^{+\infty}u_{i+1}(\mathbf{t})\partial^{-i}$, we can find
\begin{align*}
u_{2}(\mathbf{t})=Q(\mathbf{t})R(\mathbf{t})^{\intercal},\quad
u_{3}(\mathbf{t})=Q(\mathbf{t})R(\mathbf{t})^{\intercal}_{x},\quad
u_{4}(\mathbf{t})=Q(\mathbf{t})R(\mathbf{t})^{\intercal}_{xx}, \quad \cdots.
\end{align*}
In this case, $B_{2}(\mathbf{t},\partial)=\partial^{2}+2Q(\mathbf{t})R(\mathbf{t})^{\intercal}$.
Thus $Q(\mathbf{t})_{t_{2}}=B_{2}(\mathbf{t},\partial)(Q(\mathbf{t})),\, R(\mathbf{t})_{t_{2}}=-B_{2}^{*}(t\mathbf{},\partial)(R(\mathbf{t})),$
that is
\begin{align}\label{QRmatrix}
Q(\mathbf{t})_{t_{2}}=Q(\mathbf{t})_{xx}+2Q(\mathbf{t})R(\mathbf{t})^{\intercal}Q(\mathbf{t}), \quad R(\mathbf{t})=-R(\mathbf{t})_{xx}-2R(\mathbf{t})Q(\mathbf{t})^{\intercal}R(\mathbf{t}).
\end{align}
In term of matrix elements, \eqref{QRmatrix} can be expressed by
\begin{align*}
Q_{11,t_2} &= Q_{11,xx} + 2\bigl( Q_{11}\bigl(Q_{11}R_{11} + Q_{12}R_{12} + Q_{21}R_{21}\bigr) + Q_{12}Q_{21}R_{22} \bigr), \\
Q_{12,t_2} &= Q_{12,xx} + 2\bigl( Q_{12}\bigl(Q_{11}R_{11} + Q_{12}R_{12} + Q_{22}R_{22}\bigr) + Q_{11}Q_{22}R_{21} \bigr), \\
Q_{21,t_2} &= Q_{21,xx} + 2\bigl( Q_{21}\bigl(Q_{11}R_{11} + Q_{21}R_{21} + Q_{22}R_{22}\bigr) + Q_{11}Q_{22}R_{12} \bigr), \\
Q_{22,t_2} &= Q_{22,xx} + 2\bigl( Q_{22}\bigl(Q_{12}R_{12} + Q_{21}R_{21} + Q_{22}R_{22}\bigr) + Q_{12}Q_{21}R_{11} \bigr);
\end{align*}
and
\begin{align*}
-R_{11,t_2} &= R_{11,xx} + 2\bigl( R_{11}\bigl(R_{11}Q_{11} + R_{12}Q_{12} + R_{21}Q_{21}\bigr) + R_{12}R_{21}Q_{22} \bigr), \\
-R_{12,t_2} &= R_{12,xx} + 2\bigl( R_{12}\bigl(R_{11}Q_{11} + R_{12}Q_{12} + R_{22}Q_{22}\bigr) + R_{11}R_{22}Q_{21} \bigr), \\
-R_{21,t_2} &= R_{21,xx} + 2\bigl( R_{21}\bigl(R_{11}Q_{11} + R_{21}Q_{21} + R_{22}Q_{22}\bigr) + R_{11}R_{22}Q_{12} \bigr), \\
-R_{22,t_2} &= R_{22,xx} + 2\bigl( R_{22}\bigl(R_{12}Q_{12} + R_{21}Q_{21} + R_{22}Q_{22}\bigr) + R_{12}R_{21}Q_{11} \bigr).
\end{align*}

\subsection{Solutions of $(1,1)$-$2\times 2$ matrix KP hierarchy}
Here in this subsection, we will firstly construct the tau function solutions of the $(1,1)$-constrained $2$-component KP hierarchy, that is, $\tau_{\alpha\beta}({\bf x})$, $\rho_{\alpha\beta}({\bf x})$ and $\sigma_{\alpha\beta}({\bf x})$, then we can obtain the solutions of $(1,1)$-$2\times 2$ matrix KP hierarchy by
\begin{align}\label{QR}
Q_{\alpha\beta}({\bf t})=\frac{\rho_{\alpha\beta}({\bf x})}{\tau({\bf x})}\bigg|_{\mathbf{x}^{(\alpha)}=\mathbf{t}, 1\leq\alpha\leq N},\quad R_{\alpha\beta}({\bf t})=\frac{\sigma_{\alpha\beta}({\bf x})}{\tau({\bf x})}\bigg|_{\mathbf{x}^{(\alpha)}=\mathbf{t}, 1\leq\alpha\leq N}
\end{align}
where we have set $\rho_{1,\alpha\beta}$ and
$\sigma_{1,\alpha\beta}$ to be $\rho_{\alpha\beta}$ and
$\sigma_{\alpha\beta}$, respectively.
\begin{example}
Let us firstly set
$P=1, K=0, G=0, H=0, M_{11}=1, M_{12}=0$ in Theorem \ref{theorem:solution}, then we find $\mathcal{M}=1+2=3, \mathcal{N}=0$. In this case, let us assume $$\alpha_{11}=\psi^{(1)}(\lambda_{1})+\psi^{(2)}(\mu_{1}), \quad \alpha_{12}=\psi^{(1)}(\lambda_{2})+\psi^{(2)}(\mu_{2}),$$
so $\alpha_{11[1]}=\lambda_1\psi^{(1)}(\lambda_{1})+\mu_1\psi^{(2)}(\mu_{1})$.
So by Theorem \ref{theorem:solution}
\begin{align*}
&\tau_{11}(\mathbf{x})=\tau_{22}(\mathbf{x})=\langle \mathbf{p} | e^{H(\mathbf{x})}\alpha_{12}\alpha_{11[1]}\alpha_{11}|0\rangle,\\
&\tau_{12}(\mathbf{x})=\langle\mathbf{p}+\mathbf{e}_1 - \mathbf{e}_2 | e^{H(\mathbf{x})}\alpha_{12}\alpha_{11[1]}\alpha_{11}|0\rangle,\quad
\tau_{21}(\mathbf{x})=\langle\mathbf{p}+\mathbf{e}_2 - \mathbf{e}_1 | e^{H(\mathbf{x})}\alpha_{12}\alpha_{11[1]}\alpha_{11}|0\rangle,\\
&\rho_{11}(\mathbf{x})=\langle\mathbf{p}+\mathbf{e}_1 | e^{H(\mathbf{x})}\alpha_{12}\alpha_{11[2]}\alpha_{11[1]}\alpha_{11}|0\rangle,\quad \rho_{12}(\mathbf{x})=\langle\mathbf{p}+\mathbf{e}_1 | e^{H(\mathbf{x})}\alpha_{12[1]}\alpha_{12}\alpha_{11[1]}\alpha_{11}|0\rangle,\\
&\rho_{21}(\mathbf{x})=\langle\mathbf{p}+\mathbf{e}_2 | e^{H(\mathbf{x})}\alpha_{12}\alpha_{11[2]}\alpha_{11[1]}\alpha_{11}|0\rangle,\quad \rho_{22}(\mathbf{x})=\langle\mathbf{p}+\mathbf{e}_2 | e^{H(\mathbf{x})}\alpha_{12[1]}\alpha_{12}\alpha_{11[1]}\alpha_{11}|0\rangle,\\
&\sigma_{11}(\mathbf{x})=\langle\mathbf{p}-\mathbf{e}_1 | e^{H(\mathbf{x})}\alpha_{12}\alpha_{11}|0\rangle,\quad
\sigma_{12}(\mathbf{x})=\langle\mathbf{p}-\mathbf{e}_1 | e^{H(\mathbf{x})}\alpha_{11[1]}\alpha_{11}|0\rangle,\\
&\sigma_{21}(\mathbf{x})=\langle\mathbf{p}-\mathbf{e}_2 | e^{H(\mathbf{x})}\alpha_{12}\alpha_{11}|0\rangle,\quad
\sigma_{22}(\mathbf{x})=\langle\mathbf{p}-\mathbf{e}_2 | e^{H(\mathbf{x})}\alpha_{11[1]}\alpha_{11}|0\rangle.
\end{align*}
Since $\alpha_{12}\alpha_{11[1]}\alpha_{11}=(\lambda_1 - \mu_1) (\psi^{(1)}(\lambda_2)\psi^{(1)}(\lambda_1)\psi^{(2)}(\mu_1)+ \psi^{(1)}(\lambda_1)\psi^{(2)}(\mu_1)\psi^{(2)}(\mu_2))$, thus by \eqref{wicktheorem2} in Appendix, we can choose $\mathbf{p}=(1,2)$. So by the definition of $\langle\mathbf{p}|$ and \eqref{wicktheorem2}-\eqref{wicktheorem3} in Appendix, we have
\begin{align*}
&\tau_{21}(\mathbf{x})=\rho_{11}(\mathbf{x})=\rho_{21}(\mathbf{x})=\rho_{22}(\mathbf{x})=\sigma_{12}(\mathbf{x})=0,\\
&\tau_{11}(\mathbf{x})=\tau_{22}(\mathbf{x})=\mu_1(\lambda_1 - \mu_1)
e^{\xi(\mathbf{x}^{(1)},\lambda_1)
+ \xi(\mathbf{x}^{(2)},\mu_1)
+ \xi(\mathbf{x}^{(2)} - [\mu_1^{-1}],\mu_2)},\\
&\tau_{12}(\mathbf{x})=\lambda_2(\lambda_1 - \mu_1)e^{
\xi(\mathbf{x}^{(1)},\lambda_2)
+ \xi(\mathbf{x}^{(1)} - [\lambda_2^{-1}],\lambda_1)
+ \xi(\mathbf{x}^{(2)},\mu_2)},\\
&\rho_{12}(\mathbf{x})=-\lambda_2 \mu_2 (\lambda_1 - \mu_1)(\lambda_2 - \mu_2)e^{
\xi(\mathbf{x}^{(1)},\lambda_2)
+ \xi(\mathbf{x}^{(2)},\mu_2)
+ \xi(\mathbf{x}^{(1)} - [\lambda_2^{-1}],\lambda_1)
+ \xi(\mathbf{x}^{(2)} - [\mu_2^{-1}],\mu_1)},\\
&\sigma_{11}(\mathbf{x})=\mu_2 e^{\xi(\mathbf{x}^{(2)},\mu_2) + \xi(\mathbf{x}^{(2)} - [\mu_2^{-1}],\mu_1)},\\
&\sigma_{21}(\mathbf{x})= e^{\xi(\mathbf{x}^{(1)},\lambda_1) + \xi(\mathbf{x}^{(2)},\mu_2)}
- e^{\xi(\mathbf{x}^{(2)},\mu_1) + \xi(\mathbf{x}^{(1)},\lambda_2)},\\
&\sigma_{22}(\mathbf{x})=-(\lambda_1-\mu_1) e^{\xi(\mathbf{x}^{(1)},\lambda_1) + \xi(\mathbf{x}^{(2)},\mu_1)}.
\end{align*}
Finally by \eqref{QR}, we have
\begin{align*}
&Q_{11}({\bf t})=Q_{21}({\bf t})=Q_{22}({\bf t})=R_{12}({\bf t})=0,\\
&Q_{12}({\bf t})=(\lambda_2 - \lambda_1)(\lambda_2 - \mu_2) e^{\xi(\mathbf{t},\lambda_2)},\quad
R_{11}({\bf t})= -\frac{1}{\lambda_1-\mu_1}e^{-\xi(\mathbf{t},\lambda_1)},\\
&R_{21}({\bf t})=\frac{1}{(\lambda_1 - \mu_1)(\mu_1 - \mu_2)} \left(e^{-\xi(\mathbf{t},\mu_1)}-e^{\xi(\mathbf{t},\lambda_2)-\xi(\mathbf{t},\lambda_1)-\xi(\mathbf{t},\mu_2)} \right),\\
&R_{22}({\bf t})= \frac{1}{\mu_2 - \mu_1} e^{-\xi(\mathbf{t},\mu_2)}.
\end{align*}
\end{example}
\begin{example}
Next let us set $P=H=1$, $K=G=0=M_{11}=M_{12}=0$ and $N=2$ in Theorem \ref{theorem:solution}, then $\mathcal{M}=\mathcal{N}=2$. Here we will choose
\begin{align*}
\alpha_{11}=\psi^{(1)}(\lambda_{1})+\psi^{(1)}(\mu_{1}),\, \alpha_{12}=\psi^{(2)}(\lambda_{2})+\psi^{(2)}(\mu_{2}),\, \Phi^*_{11}=\psi^{*(2)}(z_{2}),\, \Phi^*_{12}=\psi^{*(1)}(z_{1}), \, \mathbf{p}=(0,0)
\end{align*}
then by Theorem \ref{theorem:solution}, we have
\begin{align*}
&\tau_{11}(\mathbf{x})=\tau_{22}(\mathbf{x})=\langle \mathbf{p} | e^{H(\mathbf{x})}\Phi^*_{12}\Phi^*_{11}\alpha_{12}\alpha_{11}|0\rangle,\\
&\tau_{12}(\mathbf{x})=\langle\mathbf{p}+\mathbf{e}_1 - \mathbf{e}_2 | e^{H(\mathbf{x})}\Phi^*_{12}\Phi^*_{11}\alpha_{12}\alpha_{11}|0\rangle,\quad
\tau_{21}(\mathbf{x})=\langle\mathbf{p}+\mathbf{e}_2 - \mathbf{e}_1 | e^{H(\mathbf{x})}\Phi^*_{12}\Phi^*_{11}\alpha_{12}\alpha_{11}|0\rangle,\\
&\rho_{11}(\mathbf{x})=\langle\mathbf{p}+\mathbf{e}_1 | e^{H(\mathbf{x})}\Phi^*_{12}\Phi^*_{11}\alpha_{12}\alpha_{11[1]}\alpha_{11}|0\rangle,\quad \rho_{12}(\mathbf{x})=\langle\mathbf{p}+\mathbf{e}_1 | e^{H(\mathbf{x})}\Phi^*_{12}\Phi^*_{11}\alpha_{12[1]}\alpha_{12}\alpha_{11}|0\rangle,\\
&\rho_{21}(\mathbf{x})=\langle\mathbf{p}+\mathbf{e}_2 | e^{H(\mathbf{x})}\Phi^*_{12}\Phi^*_{11}\alpha_{12}\alpha_{11[1]}\alpha_{11}|0\rangle,\quad
\rho_{22}(\mathbf{x})=\langle\mathbf{p}+\mathbf{e}_2 | e^{H(\mathbf{x})}\Phi^*_{12}\Phi^*_{11}\alpha_{12[1]}\alpha_{12}\alpha_{11}|0\rangle,\\
&\sigma_{11}(\mathbf{x})=\langle\mathbf{p}-\mathbf{e}_1 | e^{H(\mathbf{x})}\Phi^*_{12}\Phi^*_{11}\alpha_{12}|0\rangle,\quad
\sigma_{12}(\mathbf{x})=\langle\mathbf{p}-\mathbf{e}_1 | e^{H(\mathbf{x})}\Phi^*_{12}\Phi^*_{11}\alpha_{11}|0\rangle,\\
&\sigma_{21}(\mathbf{x})=\langle\mathbf{p}-\mathbf{e}_2 | e^{H(\mathbf{x})}\Phi^*_{12}\Phi^*_{11}\alpha_{12}|0\rangle,\quad
\sigma_{22}(\mathbf{x})=\langle\mathbf{p}-\mathbf{e}_2 | e^{H(\mathbf{x})}\Phi^*_{12}\Phi^*_{11}\alpha_{11}|0\rangle.
\end{align*}
Similarly by \eqref{wicktheorem2}-\eqref{wicktheorem3} in Appendix, we have
\begin{align*}
&\tau_{12}(\mathbf{x})=\tau_{21}(\mathbf{x})=\sigma_{12}(\mathbf{x})=\sigma_{21}(\mathbf{x})=\rho_{12}(\mathbf{x})=\rho_{21}(\mathbf{x})=0,\\
&\tau_{11}(\mathbf{x})=\tau_{22}(\mathbf{x})= e^{-\xi(\mathbf{x}^{(1)}, z_1)-\xi(\mathbf{x}^{(2)}, z_2)}
\left(
e^{\xi(\mathbf{x}^{(2)}+[z_2^{-1}], \lambda_2)}
+ e^{\xi(\mathbf{x}^{(2)}+[z_2^{-1}], \mu_2)}
\right)
\left(
e^{\xi(\mathbf{x}^{(1)}+[z_1^{-1}], \lambda_1)}
+ e^{\xi(\mathbf{x}^{(1)}+[z_1^{-1}], \mu_1)}
\right),\\
&\rho_{11}(\mathbf{x})=\frac{\lambda_1(\lambda_1-\mu_1)}{z_1}\cdot e^{-\xi(\mathbf{x}^{(1)}, z_1)-\xi(\mathbf{x}^{(2)}, z_2)+\xi(\mathbf{x}^{(1)}+[z_1^{-1}], \lambda_1)+\xi(\mathbf{x}^{(1)}+[z_1^{-1}]-[\lambda_1^{-1}], \mu_1)}
\left(
e^{\xi(\mathbf{x}^{(2)}+[z_2^{-1}], \lambda_2)}
+ e^{\xi(\mathbf{x}^{(2)}+[z_2^{-1}], \mu_2)}
\right),\\
&\rho_{22}(\mathbf{x})=-\frac{\lambda_2(\lambda_2-\mu_2)}{z_2}
\cdot e^{-\xi(\mathbf{x}^{(1)},z_1)-\xi(\mathbf{x}^{(2)},z_2)
+\xi(\mathbf{x}^{(2)}+[z_2^{-1}],\lambda_2)
+\xi(\mathbf{x}^{(2)}+[z_2^{-1}]-[\lambda_2^{-1}],\mu_2)}
\left(
e^{\xi(\mathbf{x}^{(1)}+[z_1^{-1}],\lambda_1)}
+ e^{\xi(\mathbf{x}^{(1)}+[z_1^{-1}],\mu_1)}
\right),\\
&\sigma_{11}(\mathbf{x})= z_1 \cdot e^{-\xi(\mathbf{x}^{(1)}, z_1)-\xi(\mathbf{x}^{(2)}, z_2)}
\left(
e^{\xi(\mathbf{x}^{(2)}+[z_2^{-1}], \lambda_2)}
+ e^{\xi(\mathbf{x}^{(2)}+[z_2^{-1}], \mu_2)}
\right),\\
&\sigma_{22}(\mathbf{x})=-z_2 \cdot e^{-\xi(\mathbf{x}^{(1)}, z_1)-\xi(\mathbf{x}^{(2)}, z_2)}
\left(
e^{\xi(\mathbf{x}^{(1)}+[z_1^{-1}], \lambda_1)}
+ e^{\xi(\mathbf{x}^{(1)}+[z_1^{-1}], \mu_1)}
\right).
\end{align*}
Finally by \eqref{QR},
\begin{align*}
&Q_{12}({\bf t})=Q_{21}({\bf t})=R_{12}({\bf t})=R_{21}({\bf t})=0,\\
&Q_{11}({\bf t})=\frac{(\lambda_1-\mu_1)^{2}}{(z_1-\mu_1)e^{-\xi({\bf t}, \mu_1)}
+ (z_1-\lambda_1)e^{-\xi({\bf t}, \lambda_1)}},\\
&Q_{22}({\bf t})=-\frac{(\lambda_2-\mu_2)^{2}}{(z_2-\mu_2)e^{-\xi({\bf t}, \mu_2)}
+(z_2-\lambda_2) e^{-\xi({\bf t}, \lambda_2)}},\\
&R_{11}({\bf t})=\frac{(z_{1}-\lambda_{1})(z_{1}-\mu_{1})}{(z_{1}-\mu_{1})e^{\xi({\bf t}, \lambda_1)}
+(z_{1}-\lambda_{1}) e^{\xi({\bf t}, \mu_1)}},\\
&R_{22}({\bf t})=-\frac{(z_{2}-\lambda_{2})(z_{2}-\mu_{2})}{(z_{2}-\mu_{2})e^{\xi({\bf t}, \lambda_2)}
+ (z_{2}-\lambda_{2})e^{\xi({\bf t}, \mu_2)}}.
\end{align*}
\end{example}
\section{Conclusions and discussions}
The purpose of this paper is to solve the constrained matrix KP hierarchy. The traditional method to do this is the quasi-determinants \cite{gilson2007,gelfand2005,wu2017,li2026}, but here instead we use the boson-fermion correspondence to construct solutions from the aspects of tau functions. In fact, the matrix KP hierarchy can be viewed as the reduction of the multi-component KP hierarchy by setting $\mathbf{x}^{(\alpha)}=\mathbf{t}$ for $1\leq \alpha\leq N$. Since there are tau functions for the multi-component KP hierarchy, we consider the corresponding reduction in the multi-component KP hierarchy for the constrained matrix KP hierarchy, which is called the constrained multi-component KP hierarchy, so that we can discuss the constrained matrix KP hierarchy from the aspects of tau functions. For the constrained multi-component KP hierarchy, we firstly construct the bilinear equations in terms of tau functions, and then construct the corresponding tau functions by boson-fermion correspondence. Finally by setting $\mathbf{x}^{(\alpha)}=\mathbf{t}$ for $1\leq \alpha\leq N$, we obtain the solutions for the constrained matrix KP hierarchy.

Just as we stated at the beginning of this section, as one kind of the non-commutative KP hierarchy, the matrix KP hierarchy can be solved by using the quasi-determinants. Therefore, we believe the results in this paper may be helpful in understanding non-commutative KP hierarchy in the framework of the multi-component KP hierarchy, and the relations between quasi-determinants and vacuum expectation values of free fermions.
\section*{Appendix}
In this appendix, let us give some formulas to compute the vacuum expectation value $\langle 0|a|0\rangle$ of $a\in\mathcal{A}$. Firstly given $w_1,\ldots,w_r\in V\oplus V^*$, the corresponding vacuum expectation value
\begin{align}\label{wicktheorem}
\langle 0|w_1\ldots w_r|0\rangle=
\begin{cases}
{\rm Pf}(W),&\text{if $r=2p$ for $p\geq 1$},\\
0,&\text{if $r=2p-1$ for $p\geq 1$},
\end{cases}
\end{align}
where ${\rm Pf}(W)$ is the pfaffian determinant \cite{hirota2004} of anti-symmetric matrix $W=(w_{i,j})_{1\leq i,j\leq 2p}$ given by
\begin{align*}
w_{i,j}=\begin{cases}
\langle 0|w_iw_j|0\rangle, & \text{if $i<j$,}\\
0,& \text{if $i=j$,}\\
-\langle 0|w_jw_i|0\rangle, & \text{if $i>j$.}
\end{cases}
\end{align*}
Here \eqref{wicktheorem} is referred as the so-called Wick theorem \cite{date1983,alexandrov2013}. Further notice that
\begin{align*}
\langle 0|\psi_i^{(\alpha)}\psi_j^{(\beta)}|0\rangle =\langle 0|\psi_i^{*(\alpha)}\psi_j^{*(\beta)}|0\rangle=0,\quad \langle 0|\psi_i^{(\alpha)}\psi_j^{*(\beta)}|0\rangle =\delta_{ij}\delta_{\alpha\beta} \cdot\theta(j<0),
\end{align*}
where $\theta(j>0)=1$ if $j>0$ is true, and $\theta(j>0)=0$ if $j>0$ is not true.

Based upon above facts, we have the following results. Given $w_1^{{(\gamma)}},\ldots,w_{2K_\gamma}^{(\gamma)}\in V^{{(\gamma)}}\oplus V^{{*(\gamma)}}$ ($1\leq\gamma\leq N$), we have
\begin{align}\label{wicktheorem2}
\left\langle 0\left|\prod_{\gamma=1}^N w_1^{(\gamma)}\ldots w_{2K_\gamma}^{(\gamma)} \right|0\right\rangle=\prod_{\gamma=1}^N\left\langle 0\left| w_1^{(\gamma)}\ldots w_{2K_\gamma}^{(\gamma)} \right|0\right\rangle.
\end{align}
Also for $\Theta_i^{(\gamma)}\in V^{(\gamma)}$ ($1\leq i\leq K_\gamma$) and $\Theta_j^{*(\gamma)}\in V^{*(\gamma)}$ ($1\leq j\leq L_\gamma$) that
\begin{align}\label{ptheta}
\left\langle {\bf p}\left|\prod_{\gamma=1}^N \Theta_1^{(\gamma)}\ldots\Theta_{K_\gamma}^{(\gamma)} \cdot\Theta_1^{(\gamma)*}\ldots\Theta_{L_\gamma}^{(\gamma)*} \right|0\right\rangle=0,\quad\text{if $p_\gamma\neq K_\gamma-L_\gamma$ for some $\gamma$}.
\end{align}
Besides \eqref{wicktheorem}, there are also the following determinant formulas\cite{alexandrov2013} for the vacuum expectation values
\begin{equation}
\begin{aligned}\label{wicktheorem3}
&\langle 0|\Theta_1\ldots\Theta_p\Theta_p^*\ldots\Theta_1^*|0\rangle=\det(\langle 0|\Theta_i\Theta_j^*|0\rangle)_{1\leq i,j\leq p},\\
&\langle 0|\Theta_1^*\ldots\Theta_p^*\Theta_p\ldots\Theta_1|0\rangle=\det(\langle 0|\Theta_i^*\Theta_j|0\rangle)_{1\leq i,j\leq p},
\end{aligned}
\end{equation}
where $\Theta_i\in V$ and $\Theta_{j}^*\in V^*$.

\bigskip
\noindent{\bf Acknowledgements}: \\
This work is supported by the National Natural Science Foundation of China
(Grant Nos. 12571271 and 12261072) and Huaqiao University Research Startup Funds (Grant No. 26BS117 ) .\\

\noindent{\bf Conflict of Interest}: \\
The author has no conflicts to disclose.\\

\noindent{\bf Data availability}: \\
Data sharing is not applicable to this article as no new data were created or analyzed in this study.\\

\end{document}